\documentclass[lettersize,journal]{IEEEtran}
\usepackage{amsmath,amssymb,amsfonts}

\usepackage{array}
\usepackage[caption=false,font=normalsize,labelfont=sf,textfont=sf]{subfig}
\usepackage{tikz}
\usepackage{subcaption}
\usepackage{bm}
\usepackage{textcomp}
\usepackage{stfloats}
\usepackage{url}
\usepackage{verbatim}
\usepackage{graphicx}    % For \includegraphics
\usepackage{caption}     % For \caption
\usepackage{cite}
\usepackage{amsthm}
\usepackage{wrapfig}

\newtheorem{definition}{Definition}
\newtheorem{lemma}{Lemma} 
\newtheorem{theorem}{Theorem}
\newtheorem{proposition}{Proposition}

\newtheorem{remark}{Remark}

\usepackage{algorithmic}
\usepackage[linesnumbered,ruled,vlined]{algorithm2e}

\begin{document}

\title{Mixed Traffic: A Perspective from Long Duration Autonomy}
\author{Filippos Tzortzoglou, \IEEEmembership{Student Member, IEEE}, Logan E. Beaver, \IEEEmembership{Member, IEEE}
\thanks{F. Tzortzoglou is with the Department of Civil \& Environmental Engineering, Cornell University, Ithaca, NY 11111 USA (e-mail: ft253@cornell.edu).}
\thanks{L.E. Beaver is with the Department of Mechanical \& Aerospace Engineering, Old Dominion University, Norfolk, VA 23529 USA (e-mail: lbeaver@odu.edu).}
}

\maketitle

%\thispagestyle{empty}
%\pagestyle{empty}

%%%%%%%%%%%%%%%%%%%%%%%%%%%%%%%%%%%%%%%%%%%%%%%%%%%%%%%%%%%%%%%%%%%%%%%%%%%%%%%%
\begin{abstract}
The rapid adoption of autonomous vehicle has established mixed traffic environments, comprising both autonomous and human-driven vehicles (HDVs), as essential components of next-generation mobility systems. Along these lines, connectivity between autonomous vehicles and infrastructure (V2I) is also a significant factor that can effectively support higher-level decision-making. At the same time, the integration of V2I within mixed traffic environments remains a timely and challenging problem. In this paper, we present a long-duration autonomy controller for connected and automated vehicles (CAVs) operating in such environments, with a focus on intersections where right turns on red are permitted. We begin by deriving the optimal control policy for CAVs under free-flow traffic. Next, we analyze crossing time constraints imposed by smart traffic lights and map these constraints to controller bounds using Control Barrier Functions (CBFs), with the aim to drive a CAV to cross the intersection on time. We also introduce criteria for identifying, in real-time, feasible crossing intervals for each CAV. To ensure safety for the CAVs, we present model-agnostic safety guarantees, and demonstrate their compatibility with both CAVs and HDVs. Ultimately, the final control actions are enforced through a combination of CBF constraints, constraining CAVs to traverse the intersection within the designated time intervals while respecting other vehicles. Finally, we guarantee that our control policy yields always a feasible solution and validate the proposed approach through extensive simulations in MATLAB.

\end{abstract}


\vspace{-12pt}
\section{Introduction}

\IEEEPARstart{E}{merging} mobility systems, including connected and automated vehicles (CAVs) and on-demand mobility services, present significant opportunities for creating sustainable transportation networks.  In recent decades, the global transportation sector has experienced a notable increase in automation and connectivity while CAVs have already started penetrating the market \cite{bang2024mobility}. However, 100\% penetration rate of CAVs is not expected before 2060 \cite{alessandrini2015automated}. Thus, the need for efficient protocols that allow the sustainable operation of both CAVs and human driven vehicles (HDVs) in the same transportation network is critical for their effective deployment \cite{li2023survey}. Along these lines, the long-term sustainability of such systems, without the need for periodic human feedback, is necessary for their widespread adoption. 

In the domain of mixed traffic, where both CAVs and HDVs coexist within the same transportation network, several approaches aim to address the significant challenges associated with safety and efficiency across various traffic scenarios \cite{li2024safe}. For example, recent research highlights methods for addressing the problem of efficient merging on unsignalized ramps, considering different penetration rates of autonomous vehicles \cite{le2024stochastic,chen2023deep,li2024managing}. Another study \cite{wang2022ego} examined the impact of CAVs on traffic flow under varying penetration rates, incorporating strategic lane changes. Meanwhile, some researchers have focused on signalized intersections, where CAVs plan their trajectories based on traffic signal phases to maximize the overall traffic flow at the intersection, accounting for the presence of human drivers \cite{sun2020optimal,chochliouros2021v2x}. Recent studies have also explored the simultaneous control of CAVs and traffic light phases \cite{li2024managing}. However, the joint optimization of traffic lights and CAVs remains a complex problem that is not feasible for real-time implementation, especially when number of vehicles at the intersections increases. As a result, bi-level approaches have been proposed, where CAV trajectories are determined after the time-crossing intervals are established \cite{tajalli2021traffic,tzortzoglou2025safeefficientcoexistenceautonomous}. Note that all the aforementioned studies have explored various control methods, such as optimal control \cite{le2024stochastic,sun2020optimal,tzortzoglou2025safeefficientcoexistenceautonomous}, control barrier functions \cite{li2024safe} or reinforcement learning \cite{chen2023deep,li2024managing}, each offering specific advantages and disadvantages.

Although the results from these studies are noteworthy, to the best of our knowledge, no existing approach in the literature addresses the problem of mixed traffic from the perspective of long-duration autonomy.  In this paper, we aim to bridge this gap by defining a controller that enables the holistic operation of CAVs alongside HDVs within a transportation network. Also the majority of existing studies require communication between CAVs \cite{le2024stochastic,li2024safe}. In contrast, in this study we focus exclusively on connectivity with the infrastructure (V2I), assuming that CAVs can perceive their immediate surroundings using onboard sensors, cameras, or lidar. Leveraging infrastructure-based connectivity offers a more scalable and flexible solution, while still enabling coordinated decision-making in the presence of HDVs.

By definition, a long-duration autonomy approach must be applicable across all traffic scenarios. In our view, one of the most complex scenarios involves signalized intersections that allow right turns on red traffic light. Therefore, we focus our analysis on this challenging case, considering that the proposed approach can be readily extended to other traffic scenarios. 

We first present theoretical arguments related to the optimal control of CAVs in a free-flow environment, by solving the infinite-horizon unconstrained optimal control problem. We then define conditions that ensure: 1) the safety of CAVs in the presence of other vehicles, and 2) that a CAV crosses a traffic light within a designated green phase interval. These conditions are mapped to constraints that are linear in the control input, while satisfying the forward invariance property through Control Barrier Functions (CBFs). Additionally, we analyze how a CAV can efficiently identify a feasible crossing interval to pass through the intersection. Ultimately, the final control policy is implemented by imposing state-dependent control bounds to the nominal control input through CBFs. Finally we also propose a mechanism to account for right turns on red traffic light. The proposed approach is validated through simulations in MATLAB. This article contributes to the field by: \begin{itemize} \item Providing a framework that maps infinite-horizon optimal control problems to simple reactive controllers. \item Proposing a constraint formulation that enforces green-phase intersection crossing while guaranteeing safety between CAVs and other vehicles. \item Presenting conditions for the CAVs to determine their crossing interval in real-time.  \item Introducing a framework that enables CAVs to execute right turns at red lights while accounting for oncoming vehicles. \item Guaranteeing our control policy always produces a feasible control input %has a feasible solution to our control policy.
\end{itemize}

The remainder of the article is organized as follows: In Section \ref{formulation}, we present our problem formulation. In Section \ref{CBF section}, we propose our reactive controller utilizing CBFs. In Section \ref{Handling right turns} we discuss the strategy for handling right turns on red light, and Section \ref{simulations} provides simulation results using MATLAB. Finally, in Section \ref{sec:conclusion}, we present our concluding remarks and discuss directions for future work.

\begin{figure}
        \centering
        \includegraphics[width=0.35\textwidth]{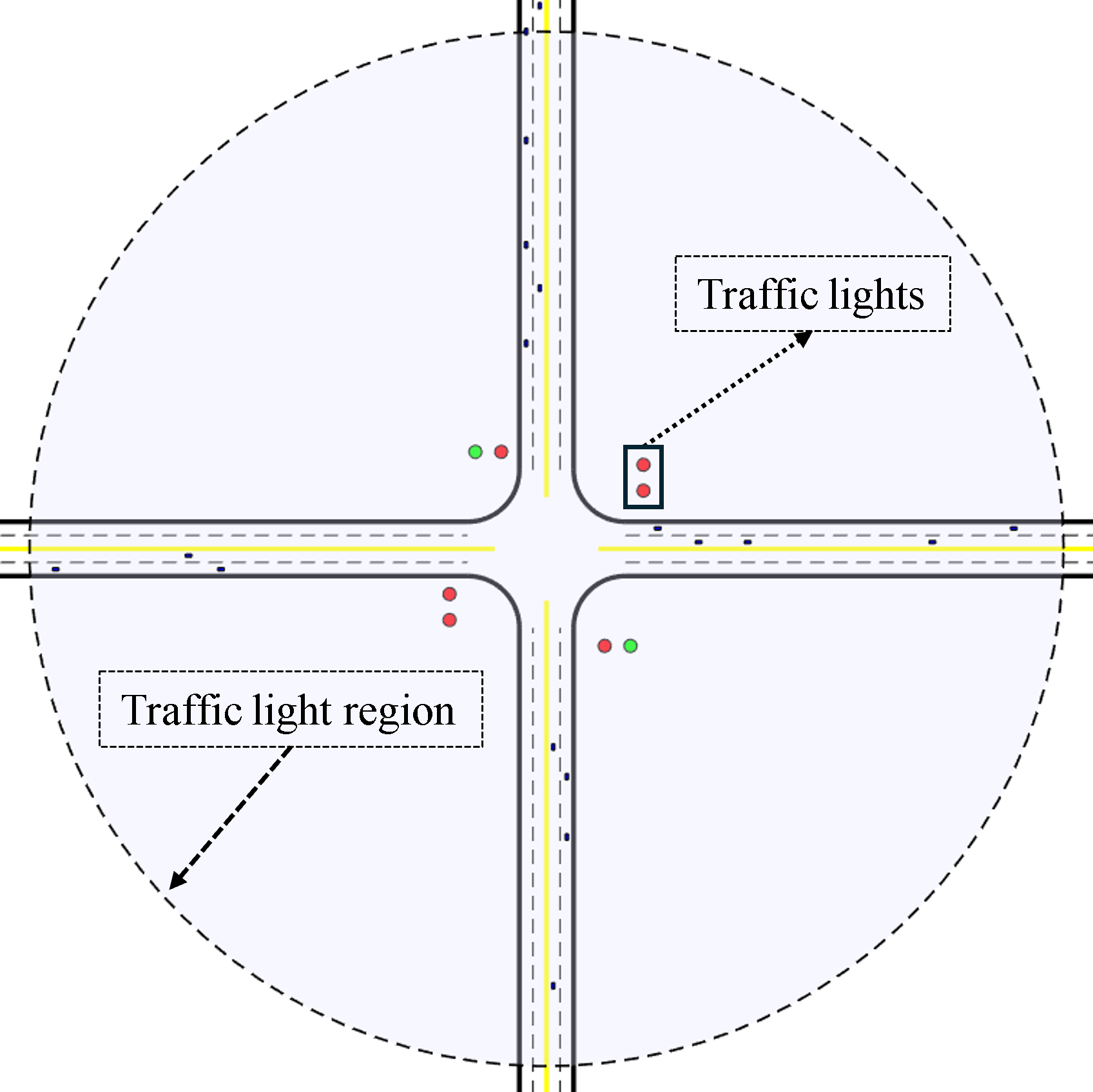}
        \caption{A four-way signalized intersection with smart traffic lights and dedicated left-turn lanes.}
    \label{fig:Intersection}
    \vspace{-12pt}
\end{figure}

% \subsection{Notation}

% Many classic references on optimal control, e.g., \cite{Bryson1975AppliedControl,Ross2015}, consider centralized optimal control problems.
% Thus, directly adopting their notation may lead to ambiguities about the state space of a decentralized problem.
% To relieve this tension, we take the following approach for an agent with index $i$.
% Endogenous variables, e.g., the position of agent $i$, are written without explicit dependence on time.
% Exogenous variables, e.g., the position of agent $j$ as measured by agent $i$, are written with an explicit dependence on time.
% This notation is common in the applied mathematics literature \cite{Levine2011OnFlatness}, and makes it explicitly clear how functions evolve with respect to the state (e.g., state dynamics) and how they evolve with respect to time (e.g., external signals measured by the agent).
\subsection{Notation}
Many classical references in optimal control (e.g., \cite{Bryson1975AppliedControl,Ross2015}) focus on centralized control problems. As a result, directly adopting their notation in decentralized settings can lead to ambiguities—particularly regarding the definition of the state space. Following the approach in \cite{beaver2024long}, we adopt the following convention for an agent indexed by \( i \). \textit{Endogenous} variables (e.g., the state of agent \( i \)) are denoted without explicit time dependence, while \textit{exogenous} variables (e.g., the position of another agent \( j \), as measured by agent \( i \)) are written as functions of time. This notation, which is common in the applied mathematics literature \cite{Levine2011OnFlatness}, clearly distinguishes between functions that evolve with the state (e.g., state dynamics) and those that evolve explicitly with time (e.g., external signals observed by the agent).

\section{Problem Formulation} \label{formulation}

Although our analysis is applicable to a wide range of traffic scenarios, we choose to focus on intersections that permit right turns on red, as we consider them among the most complex cases. This scenario effectively generalizes other challenging situations, such as roundabouts and merging, which can be viewed as variations of the same fundamental problem.  Consider now a collection of $N$ CAVs and $M$ HDVs traveling towards an intersection as in Fig. \ref{fig:Intersection}, where $N,M\in\mathbb{R}^+$. We index each vehicle with a subscript $i\in\{1,2,\dots,N,N+1,\dots,M\}$. Each vehicle has two states $p_i, v_i\in\mathbb{R}$ and a control action $u_i\in\mathbb{R}$. These correspond to the longitudinal position, speed, and acceleration of the CAV along a fixed path. The only difference between CAVs and HDVs is that in the first case the control action $u_i$ is determined based on our proposed control policy while for HDVs the control action is determined by a human driver.
To model the dynamics, we employ a double-integrator model,
\begin{equation} \label{eq:dynamics}
    \begin{bmatrix}
        \dot{p}_i \\ \dot{v}_i
    \end{bmatrix}
    = 
    \begin{bmatrix}
        0 & 1 \\ 0 & 0
    \end{bmatrix}
    \begin{bmatrix}
        p_i \\ v_i
    \end{bmatrix}
    +
    \begin{bmatrix}
        0 \\ 1
    \end{bmatrix}
    u_i,
\end{equation}
which has a control-affine representation,
\begin{equation} \label{general form of dynamics}
  \dot{\bm{x}}_i = f(\bm{x}_i,t) + g(\bm{x}_i,t)\bm{u}_i  
\end{equation}
where $\bm{x}_i=[p_i, v_i]^T$ is the state vector, $f(\bm{x}_i,t)=[v_i,0]^T$ is the drift vector field, and $g(\bm{x}_i,t)=[0,1]^T$ is the control vector field. Here, in our notation we retain the argument $t$ to conform with the general control-affine notation used later when time-varying constraints are introduced.

Finally, we define the \textit{traffic-light region} as the communication zone within which CAVs can receive information from the traffic signals (see Fig. \ref{fig:Intersection}). When a CAV enters this zone, it immediately obtains the current signal phase and remaining cycle durations, allowing it to adjust its control policy to safely cross the intersection.

Before we continue into the derivation of our control policy we impose the following assumptions throughout the paper. 

\begin{itemize}
    \item \textbf{Assumption 1:} Each CAV is equipped with low-level tracking controllers that can follow a reference speed and apply steering for lane keeping.
    \item \textbf{Assumption 2:} Communication between the CAVs and the traffic-light infrastructure is assumed to be free of appreciable noise or delay.
    \item \textbf{Assumption 3:} The traffic lights have a yellow signal that is sufficiently long for CAVs and HDVs to clear the merging zone between signal phase changes.
    \item \textbf{Assumption 4:} The traffic light has an existing signaling policy, where each phase eventually repeats.
\end{itemize}

We adopt Assumption 1 so the analysis can concentrate on higher-level decision making; deviations from perfect tracking can readily be incorporated if needed (see \cite{ChalakiCBF2022}).  
Assumption 2 can be relaxed in the same way by explicitly modeling noisy or delayed messages, as shown in \cite{chalaki2021RobustGP}. Assumption 3 is introduced primarily to prevent a clearance phase at the intersection, during which all traffic lights would be red. This assumption is not restrictive and serves only to guarantee lateral collision avoidance.
Finally, we employ Assumption 4 because designing an efficient signaling strategy is beyond the scope of this work.
Developing signaling policies for intersections is an area of active and ongoing research \cite{tzortzoglou2025safeefficientcoexistenceautonomous}.

\subsection{CAVs under free-flow}
For CAVs our first objective is the long duration deployment in urban and highway environments. As a first goal we want to track a reference speed that is either energy-optimal (e.g., maximizing fuel efficiency \cite{afdc10312}), or flow-optimal (i.e., the speed limit). In addition we regularize the objective function using the $L^2$ norm of the control input to avoid sudden jumps in acceleration, which indirectly influence fuel consumption \cite{malikopoulos2021optimal}. Thus, for each CAV $i$, we model the objective function with an infinite-horizon cost,
\begin{equation} \label{eq:Jcost}
    J(p_i, v_i, u_i) = \frac{1}{2}\int_{0}^{\infty} (v_i - v_i^d)^2 + \frac{1}{\phi^2} u_i^2\,dt,
\end{equation}
where $v_i^d$ is a desired driving speed, e.g., the energy-optimal speed or the posted speed limit, $\phi$ is a regularizing parameter to enforce smooth accelerations, and we use the fractions $1/2$ and $1/\phi^2$ to simplify the notation of our optimal solution. The cost \eqref{eq:Jcost} with dynamics \eqref{eq:dynamics} is a standard infinite-horizon Linear Quadratic Regulator problem. The optimal solution is the feedback control law \cite{beaver2024long},
\begin{equation} \label{unconstrained trajectry}
    u_i(\bm{x}_i)=\phi(v_i^d-v_i).
\end{equation}
Thus, we aim to use \eqref{unconstrained trajectry} as a reference control input for all the CAVs, which effectively \textit{continuously replans} the unconstrained trajectory based on the current speed. However, the control law in \eqref{unconstrained trajectry} corresponds to an unconstrained trajectory, meaning that it does not account for any state, control, or rear-end constraints. Moreover, in the case of an intersection, it also neglects traffic light constraints. In the subsequent analysis, we discuss how we can incorporate all these constraints into our control design, while tracking the control law \eqref{unconstrained trajectry}. We initially rigorously state all of our constraints next.

\subsection{Constraints for CAVs}

Although each CAV is subject to its physical control bounds $u_{i,\min}$ and $u_{i,\max}$ we can assume without loss of generality that all CAVs share identical control bounds. This consideration is made solely for clarity and does not affect the results. Thus for each CAV we impose,
\begin{align} \label{eq:contro_bounds}
%    0 \leq v_i& \leq v_{\max}, \\
    |u_i| \leq u_{\max},
\end{align}
Also, each CAV must respect the speed constraints,
\begin{align} \label{eq:speed_bounds}
   0 & \leq v_i \leq v_{\max},
\end{align}
and the rear-end safety constraints between sequential vehicles,
\begin{equation}\label{eq:rear-end}
    \delta_i(t)-p_i\geq \gamma_i,
\end{equation}
 where $\delta_i(t)$ measures the position of the preceding vehicle (if
 it exists), and $\gamma_i$ is a fixed standstill distance. If no preceding vehicle exists, we let $\delta_i(t) \to \infty$. Finally, each CAV must respect the traffic light constraints. That is, it must pass through a traffic light only if it is green. This constraint is represented as, 
\begin{align} \label{eq:traffic_light_bounds}
t_i^{tr}\in \mathcal{T}_i^g
\end{align}
where $t_i^{tr}$ denotes the time CAV $i$ crosses the traffic light position which is denoted as $p_i^{tr}$, while $\mathcal{T}_i^g$ denotes the set of available crossing time intervals.

At this stage, we have defined the set of constraints. Thus, our goal is to create a reactive controller based on \eqref{unconstrained trajectry} and the constraints \eqref{eq:contro_bounds}, \eqref{eq:speed_bounds}, \eqref{eq:rear-end} and \eqref{eq:traffic_light_bounds}, while ensuring forward invariance along with recursive feasibility (i.e. solution existence). In the following section, we reformulate all constraints into an equivalent form that is linear in the control input, guarantes constraint satisfaction, forward invariance and recursive feasibility.

\section{Reactive controller design} \label{CBF section}

Before proceeding with our exposition, we recall the concepts of forward invariance and class $\mathcal{K}$ functions. Conceptually, the forward invariance property ensures that if a vehicle starts in a safe state, it will remain safe for all future times. A class $\mathcal{K}$ function provides a strictly increasing, continuous mapping that quantifies how a system’s response grows with the deviation from a reference, often used to encode safety margins or robustness in control design. We now present rigorous definitions for both concepts. 
\begin{definition}
\textbf{(Forward Invariance \cite{khalil1992nonlinear})}. A set $\mathcal{C} \subset \mathbb{R}^n$ is said to be \emph{forward invariant} with respect to a dynamical system if, for every initial condition $\bm{x}(0) \in \mathcal{C}$, the solution $\bm{x}(t)$ satisfies $\bm{x}(t) \in \mathcal{C}$ for all $t \geq 0$.
\end{definition}

\begin{definition}(\textbf{Class $\mathcal{K}$ function \cite{khalil1992nonlinear})} A continuous function $\alpha: [0, \eta) \to [0, \infty)$, $\eta > 0$ is said to belong to class $\mathcal{K}$ if it is strictly increasing and $\alpha(0) = 0$.
\end{definition}

We begin by noting that each of the aforementioned constraints \eqref{eq:contro_bounds}, \eqref{eq:speed_bounds}, \eqref{eq:rear-end}, and \eqref{eq:traffic_light_bounds} can be expressed in standard CBF form $b_q(\bm{x}_i, t) \geq 0$ for all $i \in \mathcal{N}$, where $q$ denotes the constraint index. Accordingly, the safe set associated with constraint $q$ and vehicle $i$ corresponds to a super-level set $\mathcal{C}_{i,q}(t) = \{\,\bm{x}_i \in \mathbb{R}^2 : b_q(\bm{x}_i,t) \ge 0\,\}.$ 

\begin{theorem}[Following \cite{Lindemann2019ControlTasks}]
    Forward invariance of a set $\mathcal{C}_{i,q}$ with respect to constraint $b_q(\bm{x}_i,t) \ge 0$ is guaranteed when \begin{equation}
    \label{eq:time_dep_cbf}
    \mathcal{L}_f b_q(\bm{x}_i, t) + \mathcal{L}_g b_q(\bm{x}_i, t) u_i + \frac{\partial b_q}{\partial t}(\bm{x}_i, t) \geq -\alpha(b_q(\bm{x}_i, t)),
\end{equation}
where $\mathcal{L}_f b_q$ and $\mathcal{L}_g b_q$ denote the Lie derivatives of $b_q$ with respect to $f$ and $g$, respectively.
\end{theorem}

\begin{proof}
Using the chain rule, the total derivative of $b_q(\bm{x}_i, t)$ is given by
\begin{equation} \label{eq:dbdt}
    \frac{d}{dt}b_q(\bm{x}_i, t) = \nabla_{\bm{x}_i} b_q(\bm{x}_i,t) \cdot \dot{\bm{x}}_i + \frac{\partial b_q}{\partial t}(\bm{x}_i,t).
\end{equation}
Substituting the control-affine dynamics \eqref{general form of dynamics} into \eqref{eq:dbdt}, we obtain
\begin{align} \label{dot_b}
    \dot{b}_q(\bm{x}_i, t) =& 
    \nabla_{\bm{x}_i} b_q(\bm{x}_i, t) \cdot f(\bm{x}_i, t) + \notag \\ &\nabla_{\bm{x}_i} b_q(\bm{x}_i, t) \cdot g(\bm{x}_i, t) u_i + \frac{\partial b_q}{\partial t}(\bm{x}_i, t).
\end{align}
Note that \eqref{dot_b} can be written using Lie derivative notation,
\begin{equation}
    \dot{b}_q(\bm{x}_i, t) = \mathcal{L}_f b_q(\bm{x}_i, t) + \mathcal{L}_g b_q(\bm{x}_i, t) u_i + \frac{\partial b_q}{\partial t}(\bm{x}_i, t).
\end{equation}
By the Nagumo theorem \cite{Ames2019ControlApplications}, the set $\mathcal{C}_{i,q}(t) = \{\bm{x}_i \in \mathbb{R}^2 : b_q(\bm{x}_i, t) \ge 0\}$ is forward invariant if the inequality
\begin{equation}
\dot{b}_q(\bm{x}_i, t) \ge 0
\end{equation}
holds for all $\bm{x}_i$ on the boundary of the constraint set, i.e., where $b_q(\bm{x}_i, t) = 0$.
However, given that the Nagumo theorem applies only on the boundary, we want to ensure that the control input proactively avoids approaching the boundary. Thus, we introduce a class-$\mathcal{K}$ function $\alpha(\cdot)$ and impose the condition:
\begin{equation} \label{dot b with K function}
\dot{b}_q(\bm{x}_i, t) \ge -\alpha(b_q(\bm{x}_i, t)).
\end{equation}
Substituting \eqref{dot_b} into \eqref{dot b with K function} and using Lie derivative notation, yields the barrier condition
\begin{equation}
 %   \label{eq:time_dep_cbf}
    \mathcal{L}_f b_q(\bm{x}_i, t) + \mathcal{L}_g b_q(\bm{x}_i, t) u_i + \frac{\partial b_q}{\partial t}(\bm{x}_i, t) \geq -\alpha(b_q(\bm{x}_i, t)),
\end{equation}
which satisfies the Nagumo theorem and is thus forward invariant.
\end{proof}

\begin{remark}
The Lie derivative of a scalar function $b(\bm{x})$ with respect to a vector field $f$ is given by $\mathcal{L}_f b(\bm{x}) = \nabla b(\bm{x}) f(\bm{x})$, where $\nabla b(\bm{x})$ is the gradient of $b(\bm{x})$. 
\end{remark}

% Next, as it was shown in \cite{ames2016control,ames2019control}, we can transform a constraint of the form $b_q(\bm{x}(t))\geq0$ onto a new equivalent constraint that is linear in the control input and guarantees the forward invariance property by applying the following inequality
% \begin{equation}    \mathcal{L}_fb_q(\bm{x}(t))+\mathcal{L}_gb_q(\bm{x}(t))u(t)+\alpha(b_q(\bm{x}(t)))\geq 0,
%     \label{CBF constraint}
% \end{equation}

Given \eqref{eq:time_dep_cbf}, we can transform the speed bounds \eqref{eq:speed_bounds} and traffic light constraints \eqref{eq:traffic_light_bounds} into the form of a time-varying barrier condition \eqref{eq:time_dep_cbf}.
Note that for the rear-end constraint \eqref{eq:rear-end}, the condition \eqref{eq:time_dep_cbf} does not yield an inequality that explicitly involves the control input $u$. Therefore, we must incorporate the concept of a High Order Control Barrier Function (HOCBF), which enables us to take successive derivatives of the constraint until the control input appears explicitly.
Finally, the control bounds \eqref{eq:contro_bounds} are already linear in the control input and thus require no further analysis.

% Recall that the rear end safety constraint is $p_k(t)-p_i(t)\geq \phi v_i(t) + \gamma$. Thus, we can define $b_1(\bm{x}_i(t))=p_k(t)-p_i(t)-\gamma -\phi v_i(t)$, and by that we can represent the rear-end constraint as:
% \begin{equation}
% b_1(\bm{x}_i(t))\geq 0.
% \end{equation}
% Then, according to \eqref{CBF constraint} we obtain:

% \begin{equation}
% \underbrace{v_k(t) - v_i(t)}_{\mathcal{L}_{f}b_1} + \underbrace{ - \phi}_{\mathcal{L}_{g}b_1} u_i(t) + \underbrace{p_k(t) - p_i(t) - \gamma - \phi v_i(t)}_{\alpha (b_1)=b_1} \geq 0.
%     \label{CBFrearend1}
% \end{equation}

\subsection{Barrier function design}
The speed limits of the vehicles can be written as $b_1(\bm{x}_i)=v_{\text{max}}-v_i\geq 0$ and $b_2(\bm{x}_i)=v_i-v_{\text{min}} \geq 0$. Substituting $b_1(\bm{x}_i)$ and $b_2(\bm{x}_i)$ into \eqref{eq:time_dep_cbf}, and selecting a linear class $\mathcal{K}$ function with coefficient $\kappa_s$, and rearranging terms, yields the barrier conditions,
\begin{align} 
    u_i &\leq \kappa_s(v_{\text{max}}-v_i \label{speed_CBF_1})\\
     u_i &\geq \kappa_s(-v_i+v_{\text{min}} \label{speed_CBF_2}). 
\end{align}
As shown in \cite{beaver2021constraint}, letting $\kappa_s \to \infty$ maximizes the size of the safe sets $\mathcal{C}_{i,1}$ and $\mathcal{C}_{i,2}$, thereby minimizing the impact of the constraint. Conversely, smaller values of $\kappa_s$ lead to more conservative constraints, meaning that a CAV proactively adjusts its control bounds even before approaching the speed limits.

Next, we consider the traffic light crossing constraint \eqref{eq:traffic_light_bounds}. Let us define $t_i^0$ as the time at which a CAV enters the traffic light region to cross the upcoming intersection. Given a set of admissible crossing intervals $\mathcal{T}_i^g = \{[t_i^1, t_i^1], [t_i^2, t_i^2], \dots, [t_1^{k-1}, t_i^k]\}$, our objective is to ensure that the CAV crosses the traffic light within one of these intervals.
However, two key questions arise:  
1) How can we drive a CAV to cross the traffic light during a selected interval? 2) How does a CAV select a time interval from a set of disjoint available intervals? We address the first question in the following analysis and the second one in Section \ref{crossing interval definition}. We present the results for a crossing time CBF next.

\begin{lemma} \label{lma:primitives}
    The conditions
    \begin{align}
        v_i \geq& \frac{\Delta p}{\Delta t} - \frac{1}{2} u_{\max} \Delta t, \label{eq:kin-11} \\
        v_i \leq& \frac{\Delta p}{\Delta t} + \frac{1}{2} u_{\max} \Delta t, \label{eq:kin-2}
    \end{align}
    guarantee that CAV $i$ travels a distance $\Delta p$ before (after) the time interval $\Delta t$ elapses.
\end{lemma}

\begin{proof}
    The proof comes from re-arranging the constant-acceleration kinematic equations; see \cite{beaver2024optimal}.
\end{proof}

Lemma 1 offers a valuable insight as it enables us to formulate a condition that links the vehicle's current speed, its acceleration bounds, a given distance $\Delta p$, and a time interval $\Delta t$. This condition ensures that the CAV can successfully traverse the specified distance within the allotted time. Thus, we can take advantage of this result and map the constraints \eqref{eq:kin-11} and \eqref{eq:kin-2} into barrier conditions using \eqref{eq:time_dep_cbf}. This allows us to derive control bounds that drive each CAV to cross the traffic light within a specified time interval, which we prove in the following result.

\begin{theorem} \label{thm:main}
Let CAV $i$ approach a traffic light at position $p_i^{tr}$ and consider the crossing time interval $[\underline{t}_i^{k}, \bar{t}_i^{k}]$. Then, the barrier conditions,
\begin{equation}
    u_i \geq \kappa_T \left( \frac{\Delta p_i}{\Delta t_{i,2}} - u_{\max} \frac{\Delta t_{i,2}}{2} - v_i \right)
    + \frac{\Delta p_i - v_i \Delta t_{i,2}}{\Delta t_{i,2}^2} + \frac{u_{\max}}{2},
    \label{boundary1}
\end{equation}
\begin{align}
    u_i \leq -\kappa_T \left( v_i - \frac{\Delta p_i}{\Delta t_{i,1}} - u_{\max} \frac{\Delta t_{i,1}}{2} \right)&+\frac{\Delta p_i - v_i \Delta t_{i,1}}{\Delta t_{i,1}^2} \nonumber \\&-\frac{u_{\max}}{2},
    \label{boundary2}
\end{align}
\noindent where $\Delta p_i := p_i^{tr} - p_i$, $\Delta t_{i,1} := \underline{t}_i^{k} - t$, $\Delta t_{i,2} := \bar{t}_i^{k} - t$, and $\kappa_T > 0$ is a gain, guarantee the arrival of CAV $i$ at position $p_i^{tr}$ during the interval $[\underline{t}_i^{k}, \bar{t}_i^{k}]$.
\end{theorem}
\begin{proof}
We begin with the inequality \eqref{boundary1}. Note that the condition \eqref{eq:kin-11} can be equivalently rewritten as
$b_3(\bm{x}_i,t) := v_i - \frac{\Delta p}{\Delta t} + \frac{1}{2} u_{\max} \Delta t \geq 0$. We define a linear class $\mathcal{K}$ function $\alpha(b) := \kappa_T b$ with gain $\kappa_T > 0$. We apply the time-dependent barrier condition \eqref{eq:time_dep_cbf} to $b_3(\bm{x}_i,t)$, which guarantees set invariance by Theorem 1. This yields inequality \eqref{boundary1}. Correspondigly, for inequality \eqref{boundary2}, we apply again \eqref{eq:time_dep_cbf} to the condition \eqref{eq:kin-2}, which can be rewritten as $b_4(\bm{x}_i,t) := \frac{\Delta p}{\Delta t} - \frac{1}{2} u_{\max} \Delta t - v_i \geq 0$. This yields the result.
\end{proof}

Theorem 2 provides a significant result, as it enables the design of barrier conditions associated with available time intervals and the corresponding points along the CAV trajectories. As previously discussed, condition \eqref{boundary1} defines a lower bound on the control input of a CAV to ensure it crosses the point \( p_i^k \) no later than the time moment \( \bar{t}_i^{k} \). Similarly, condition \eqref{boundary2} establishes an upper bound, ensuring that the CAV crosses the point \( p_i^k \) at time \( \underline{t}_i^{k} \) or later. 

Finally, to create a barrier condition for rear-end safety constraint \eqref{eq:rear-end}, we again utilize \eqref{eq:time_dep_cbf}. Here we select the square root as our class-$\mathcal{K}$ function. That is $\alpha(b_q(\bm{x}_i,t))=\sqrt{2u_{\max}b_q(\bm{x}_i,t)}$.
This choice of class-$\mathcal{K}$ function leads to the classic minimum stopping distance constraint \cite{beaver2021constraint} after applying
\eqref{eq:time_dep_cbf},
\begin{equation} \label{rear_end_first_order_cbf}
    \dot{\delta_i}(t)-v_i+\sqrt{2u_{\max}(\delta_i(t)-p_i-\gamma_i)}\geq0.
\end{equation}
Note that by construction, the safe safe associated with \eqref{rear_end_first_order_cbf} will always respect \eqref{eq:contro_bounds} \cite{beaver2021constraint}. However, as previously discussed \eqref{rear_end_first_order_cbf} does not incorporate the control input $u_i$ and thus requires a HOCBF. To achieve this, we treat $\dot{\delta_i}(t)-v_i+\sqrt{2u_{\max}(\delta_i(t)-v_i-\gamma_i)}\geq0$ as a new constraint and we apply \eqref{eq:time_dep_cbf} to \eqref{rear_end_first_order_cbf} with a linear class $\mathcal{K}$ function with coefficient $\kappa_R$. This yields
\begin{align}
&\ddot{\delta}_i(t)-u_i+\frac{u_{\text{max}}( \dot{\delta}_i(t)-v_i)}{\sqrt{2u_{\text{max}}(\delta_i(t)-p_i -\gamma)}} \\ \nonumber
&+\kappa_R(\dot{\delta_i}(t)-v_i+\sqrt{2u_{\max}(\delta_i(t)-v_i-\gamma_i)}) \geq 0,
\end{align}
and rearranging terms yields,
\begin{align} \label{CBF_rear_end}
&u_i\leq\ddot{\delta}_i(t)-\frac{u_{\text{max}}( v_i-\dot{\delta}_i(t))}{\sqrt{-2u_{\text{max}}(p_i-\delta_i(t) +\gamma)}} \\ \nonumber
&-\kappa_R(v_i-\dot{\delta_i}(t)+\sqrt{-2u_{\max}(v_i-\delta_i(t)+\gamma_i))}.
\end{align}

\noindent Note that in \eqref{CBF_rear_end}, the term $\ddot{\delta}_i(t)$ represents the acceleration of the preceding vehicle. If the preceding vehicle is a CAV, we can assume that the following vehicle has direct access to the value of $\ddot{\delta}_i(t)$. On the other hand, if the preceding vehicle is an HDV, we can either assume that its acceleration is constant—a common assumption in MPC controllers in mixed traffic environments (see \cite{le2024distributed}) or, for a more robust approach, we can substitute $\ddot{\delta}_i(t)$ with a value lower than or equal to $-u_{\max}$, representing the minimum deceleration an HDV can achieve.

\begin{remark}
Note that if we assume a constant acceleration model for the HDVs when computing \( \ddot{\delta}_i(t) \) in \eqref{CBF_rear_end}, we may not precisely capture their actual behavior. However, since our controller is reactive and updates its control inputs at each time step, we expect our algorithm to adapt at least as well as any MPC-based controller. %this implementation can effectively handle any discrepancies.
\end{remark}

Note that $\delta_i(t)$ and its derivatives may be discontinuous, i.e., if another CAV turns into the lane in front of CAV $i$ or CAV $i$ turns to a different lane that has a preceding vehicle. Recall that in the event that no CAV precedes $i$, we let $\delta_i(t)~\to~\infty$, and \eqref{CBF_rear_end} is trivially satisfied.

\subsection{Crossing interval selection}\label{crossing interval definition}
In the previous section we derived barrier conditions \eqref{boundary1} and \eqref{boundary2} that can drive a CAV to cross the traffic light (i.e. travel the distance $\Delta p_i)$ during a given time interval $\Delta t$. However, although the derivation of $\Delta p_i$ is trivial given that the traffic light is at a fixed point along the trajectory, the selection of a crossing time interval  $[\underline{t}_i^{k}, \bar{t}_i^{k}]$ from the set $\mathcal{T}_i^g$ needs careful handling since it can have a significant impact leading to long wait times or infeasible solutions. This problem refers to the second question we stated earlier regarding the selection of the green interval within which each CAV must cross. 

We now provide conditions for the online determination of the interval $[\underline{t}_i^{k}, \bar{t}_i^{k}]\in \mathcal{T}_i^g$. 
As it has been already discussed, upon entering the traffic light zone, a CAV receives information regarding the available green phase intervals from the traffic light system. To maximize traffic throughput, we choose to prioritize the earliest available time interval within the provided schedule. However, sometimes the earliest time interval might not be achievable due to the speed and acceleration limits of the CAV. Thus, to guarantee that the selection corresponds to a feasible CAV trajectory, we introduce the following results.
\begin{lemma} \label{traffic_feasibility}
When CAV $i$ enters the control zone, the signal crossing time determined by the light must satisfy,
    \begin{equation} \label{condition1}
        \Delta t_{i,2} \geq \frac{\sqrt{v_i^2 + 2 u_{\max}\Delta p_i} - v_i}{u_{\max}}.
    \end{equation}
\end{lemma}
\begin{proof}
    Let us define the safe set associated with constraint \eqref{eq:kin-11} as $\mathcal{C}_{i,4}$. From Definition 1, we need $\bm{x}_i(0)\in \mathcal{C}_{i,4}$, i.e.,
    \begin{equation} \label{eq:lma-quad}
        \frac{1}{2}u_{\max}\Delta t^2 + v_i\Delta t - \Delta p \geq 0.
    \end{equation}
    Solving \eqref{eq:lma-quad} for $\Delta t$ (in our case $\Delta t_{i,2}$) provides the critical times for constraint satisfaction at entry,
    \begin{equation}
        \Delta t^* = \frac{- v_i\pm \sqrt{v_i^2+2u_{\max}\Delta p_i}}{u_{\max}}.
    \end{equation}
    By the Vieta condition the values of $\Delta t^*$ have opposite signs, and we take the positive solution.
\end{proof}

The condition in Lemma \ref{traffic_feasibility} is initially checked so that CAV~$i$ can determine a feasible crossing interval. However, due to the dynamic nature of the state of the intersection, condition \eqref{condition1} must be continuously monitored until CAV~$i$ completes its traversal of the intersection. If, at any point, the initially selected time interval $\Delta t_{i,2}$ fails to satisfy the condition in Lemma~\ref{traffic_feasibility} (e.g., because the preceding vehicle is traveling slowly), then the vehicle must re-evaluate its decision and select a later feasible crossing interval. In the following Lemma we prove that there will be always a crossing time interval satysfing condition \eqref{condition1}.
\begin{lemma}\label{existence_of_feasible_intervals}
Under Assumption 4, for each CAV $i$ entering the control zone there exists a time interval $[\underline{t}_i^k, \bar{t}_i^k]$ such that condition \eqref{condition1} is satisfied.    
\end{lemma}
\begin{proof}
Let CAV $i$ enter the control zone at time $t_i^0$ with initial speed $v_i$, and let $\Delta p_i$ denote the remaining distance to the traffic light. Note that the values $v_i$ and $\Delta p_i$ are fixed upon entry, and the CAV must respect $u_{\max}$ by \eqref{eq:contro_bounds}.
Thus, the right hand side of \eqref{condition1} is fixed. Hence, increasing $\Delta t_{i,2}$ by updating the crossing time interval, will eventually satisfy \eqref{condition1}.    
\end{proof}
\begin{remark}
The result of Lemma~\ref{existence_of_feasible_intervals} holds for any time instant and is not restricted to the moment a CAV enters the control zone.
\end{remark}

Up to this point, we have defined the barrier conditions \eqref{boundary1} and \eqref{boundary2}, which ensure that a CAV traverses the intersection within a specified time interval. We have also introduced condition \eqref{condition1}, which determines whether a given crossing time interval is feasible.

However, a significant note that must be underlined is related to the nature of barrier condition \eqref{boundary2}. Note that \eqref{boundary2} implies that CAV $i$ will travel a distance $\Delta p_i$  at time $\Delta t_{i,1}$ or later. However, it does not directly consider the minimum speed constraint \eqref{eq:speed_bounds}. That is \eqref{boundary2} does not strictly guarantee that the CAV does not overshoot $\Delta p_i$, apply negative control input, come to a stop, and come back in reverse to reach distance $\Delta p_i$ at time $\Delta t_{i,1}$ or later.
We avoid this pathological case with our next result.

\begin{theorem} \label{thm:no-overshoot}
For a crossing time interval $[\underline{t}_i^{k}, \bar{t}_i^{k}]$ and CAV $i$, if  $\underline{t}_i^{k}$ satisfies \begin{equation} \label{eq:cross-time-limit}
    \Delta t_{i,1} \leq \sqrt{\frac{2 \Delta p_i}{u_{\max}}},
\end{equation}
\noindent then,  CAV $i$ will not cross the intersection before $\Delta t_{i,1}$ elapses.
\end{theorem}
\begin{proof}

Condition \eqref{boundary2} guarantees satisfaction of \eqref{eq:kin-2},
\begin{equation} \label{helper_theorem_3}
    v_i \leq \frac{\Delta p_i}{\Delta t_{i,1}} + \frac{1}{2} u_{\max} \Delta t_{i,1}. 
\end{equation}
To avoid overshooting the intersection, it is sufficient to impose the condition,
\begin{equation}\label{positive_speed}
    v_i \geq 0, \quad \forall t\in[t_i^0, \underline{t}_i^k].
\end{equation}
Thus, we want to ensure that the maximum braking input $-u_{\max}$ applied over the interval $\Delta t_{i,1}$ does not cause the vehicle to move in reverse. Thus, we seek to satisfy,
\begin{align} \label{helper2_theorem_3}
    0 &\leq v_i - u_{\max}\Delta t_{i,1}. 
    \end{align}
    Substituting \eqref{helper_theorem_3} into \eqref{helper2_theorem_3} yields,
    \begin{align}
    0 \leq \frac{\Delta p_i}{\Delta t_{i,1}} + \frac{1}{2} u_{\max} \Delta t_{i,1} - u_{\max}\Delta t_{i,1}.
\end{align}
Thus, for \eqref{positive_speed} to be satisfied, the following must hold,
\begin{equation}
    \frac{\Delta p_i}{\Delta t_{i,1}} - \frac{1}{2} u_{\max} \Delta t_{i,1} \geq 0.
\end{equation}
Solving for $\Delta t_{i,1}$ completes the proof. 

\end{proof}

Intuitively, Theorem \ref{thm:no-overshoot} gives a condition that $\Delta t_{i,1}$ must be ``small enough'' so that CAV $i$ does not have sufficient time to overshoot the intersection. 
Formally, when braking at the maximum rate, \eqref{eq:kin-2} describes the speed of agent $i$ with the differential equation,
\begin{equation} \label{eq:v-ode}
    \dot{p}_i = \frac{\Delta p_i}{\Delta t_{i,1}}  - \frac{1}{2}u_{\max} \Delta t_{i,1},
\end{equation}
which is not continuous at $\Delta t_{i,1}\to 0$.
Thus, \eqref{eq:v-ode} may admit multiple particular solutions, and Theorem \ref{thm:no-overshoot} restricts us to the set of initial states that guarantees our solution doesn't overshoot the distance $\Delta p_i$.

% To avoid this, we want to simultaneously impose that
% \begin{equation}
%     0 \leq v(t_1) = v(t) - u_{\max}\Delta t.
% \end{equation}
% This gives an upper and lower bound on $v$,
% \begin{equation}
%     u_{\max}\Delta t \leq v \leq \frac{\Delta p_i}{\Delta t} - \frac{1}{2} u_{\max}\Delta t.
% \end{equation}
% Simplifying yields the consistency constraint,

%The upper bound on $\Delta t_{i,1}$ derived in \eqref{eq:cross-time-limit} guarantees that the vehicle stops at the traffic light without first overshooting it. In the event that condition \eqref{eq:cross-time-limit} is violated, this undesirable behavior must be explicitly prevented. To that end, we enforce the rear-end safety constraint \eqref{eq:rear-end} at the entrance of the intersection during the red signal phase. The inclusion of this constraint is straightforward and can be achieved by introducing a virtual vehicle positioned at the entry point of the intersection, thereby ensuring that the approaching vehicle maintains a safe distance and does not proceed past the stop line prematurely.

Theorem \ref{thm:no-overshoot} gives an upper bound on $\Delta t_{i,1}$ that is too restrictive for most cases.
Instead, we use this result to determine the braking behavior of CAV $i$.
Namely, if \eqref{eq:cross-time-limit} is satisfied, the CAV brakes as normal.
Otherwise, we relax \eqref{boundary2} and use the rear-end constraint \eqref{eq:rear-end} to stop the vehicle, i.e., we place a virtual stopped vehicle at the intersection. 
This switching behavior is depicted in Fig. \ref{fig:switching}.

\begin{figure}[ht]
    \centering
    \begin{tikzpicture}[
        block/.style={draw,minimum width=3cm,thick},
        arrow/.style={->,thick}]
        \node[block] (KIN) at (0,0) {Rear-end brake \eqref{CBF_rear_end}};
        \node[block] (REA) at (0,-1) {Cros. time brake \eqref{boundary2}};
        \draw[arrow] (KIN.0) to [bend left] node[right]{$\Delta t_{i,1} \leq \sqrt{\frac{2 \Delta p_i}{u_{\max}}}$} (REA.0);
        \draw[arrow] (REA.180) to [bend left] node[left]{$\Delta t_{i,1} > \sqrt{\frac{2 \Delta p_i}{u_{\max}}}$} (KIN.180);
    \end{tikzpicture}
    \caption{Braking behavior for the CAV approaching the stop light.}
    \label{fig:switching}
\end{figure}

Note that chattering between the two braking formulas is undesirable, as it can lead to continuous fluctuations in the control bounds. Next, we show this switching behavior does not lead to any chattering between the two braking formulas \eqref{CBF_rear_end} and \eqref{boundary2}.

\begin{proposition} \label{prp:chatter}
   The switch from crossing time stopping to rear-end stopping occurs at most once.
\end{proposition}

\begin{proof}
    For switching to occur, the time $\Delta t_{i,1}$ must satisfy,
    \begin{equation} \label{eq:zeno-ineq}
        \Delta t_{i,1} - \sqrt{\frac{2\Delta p_i}{u_{\max}}} > 0,
    \end{equation}
    i.e., it is sufficiently large to violate Theorem \ref{thm:no-overshoot}. When \eqref{eq:zeno-ineq} is satisfied, we consider a virtual stopped vehicle at the intersection. So in \eqref{rear_end_first_order_cbf}, we can substitute $\delta_i(t)-p_i-\gamma_i$ with $\Delta p_i$. Hence, \eqref{eq:speed_bounds} and \eqref{rear_end_first_order_cbf} imply that the speed of the CAV is bounded by,
    \begin{equation}\label{helper}
        0 \leq \dot{p} \leq \sqrt{2u_{\max}\Delta p_i}.
    \end{equation}
    Next note that the derivative of $\Delta t_{i,1} - \sqrt{\frac{2\Delta p_i}{u_{\max}}}$ is given by,
        \begin{equation} \label{eq:sandwich_derivative}
        - 1+\dot{p}\,\sqrt{\frac
        {1}{2u_{\max}\Delta p_i}.
        }
    \end{equation}
    Using \eqref{helper} to bound $\dot{p}$ in  \eqref{eq:sandwich_derivative} yields,
    \begin{equation} \label{eq:sandwich}
        -1 \leq -1 + \dot{p}\,\sqrt{\frac
        {1}{2u_{\max}\Delta p_i}
        } 
        \leq
        \sqrt{1} - 1 = 0,
    \end{equation}
    which implies that the LHS of \eqref{eq:zeno-ineq} is non increasing as $\Delta p_i \to 0$.
    At the limit $\Delta p_i = 0$, $\dot{p} = 0$ by \eqref{helper}.
    Thus the vehicle remains stopped until $\Delta t_{i,1}$ elapses.
\end{proof}

Note that all the conditions defined thus far for determining the crossing time interval do not explicitly account for the satisfaction of speed constraints. 
In particular, \eqref{boundary1} may generate a lower bound that requires $v_i > v_{\max}$.
This would lead the CAV to accelerate and reach $v_{\max}$, the condition in Lemma 2 would eventually become false, and the CAV must then request a later stopping time and brake.
To avoid this situation, we impose another condition on the feasible crossing time, which we present next.
\begin{lemma} \label{lma:vmax}
    When neglecting HDVs, for a crossing time interval $[\underline{t}_i^k, \overline{t}_i^k]$, the condition
    \begin{equation}
        \Delta t_{i,2} \geq 
        \frac{2 \Delta p_i }{v_{\max}+v_i}
    \end{equation}
    where $\Delta t_{i,2} = \overline{t}_i^k - t$,
    is sufficient to ensure the CAV can reach the intersection while respecting the speed bounds \eqref{eq:speed_bounds}.
\end{lemma}

\begin{proof}
Consider a CAV $i$, the current state $\Delta p_i, v_i$, and the kinematic relation,
\begin{equation}
\Delta p_i = v_i \Delta t_{i,2} + \frac{1}{2} u (\Delta t_{i,2})^2,
\label{eq:kinematic_relation1}
\end{equation}
that determines the constant acceleration $u$ required to travel a distance $\Delta p_i$ by time $\Delta t_{i,2}$.
Solving \eqref{eq:kinematic_relation1} for $u$ yields,
\begin{equation} \label{eq:u-bound-vmax}
    u = 2 \frac{\Delta p_i - v_i\Delta t_{i,2}}{(\Delta t_{i,2})^2},
\end{equation}
which satisfies $u \leq u_{\max}$ by construction of $\Delta t_{i,2}$ (see \eqref{eq:kin-2}).
For $u \leq 0$, the following speed trajectory is decreasing and the result holds trivially.
For $u > 0$, the maximum speed occurs at the crossing; thus it is necessary and sufficient to impose,
\begin{equation}
    v_{\max} \geq v_f = v_0 + u \Delta t_{i,2}.
\end{equation}
Rearranging and substituting \eqref{eq:u-bound-vmax} completes the proof.
%
%\eqref{eq:kinematic_relation1} for $u$.

%To ensure the speed constraint is satisfied throughout the maneuver, we require that the resulting final velocity,
%\begin{equation}
%v_{\text{final}} = v + u \Delta t,
%\label{eq:final_velocity}
%\end{equation}
%does not exceed the maximum allowable speed $v_{\max}$. That is, the following condition must hold:
%\begin{equation}
%v + u \Delta t \leq v_{\max}.
%\label{eq:speed_constraint}
%\end{equation}%

\end{proof}

Note that Lemma \ref{lma:vmax} is conservative, as it only guarantees the existence of a trajectory that crosses the intersection without reaching $v_{\max}$.
In particular, we exclude the signal phases that \textit{require} traveling at $v=v_{\max}$ for crossing to be feasible.
While this could be relaxed by explicitly including the speed bounds in Lemma \ref{lma:primitives}, we expect that, for mixed traffic at a crowded intersection, there will inevitably be delays that will require the CAV to request a later stopping time and quickly brake.

At this stage, we are equipped to present our optimization problem along with the complete control strategy for the CAVs, which is defined in Algorithm~1. Following this, we prove that our proposed control policy is always feasible.
\\
\textbf{Problem 1:} At each time, each CAV $i$ applies the control action
that satisfies,
\begin{align}\label{Time_optimal}
 \underset{u_i}{\min}& \quad \bigg(\phi(v_i^d-v_i)-u_i)\bigg)^2  \nonumber\\
\text{subject to}& \nonumber\\ \nonumber
\text{control bounds:}&\;\eqref{eq:contro_bounds},\nonumber\\
\text{speed constraints:}&\;\eqref{speed_CBF_1},\;\eqref{speed_CBF_2}, \nonumber\\ 
\text{rear-end constraints:}&\;\eqref{CBF_rear_end},\nonumber\\ 
\text{crossing time constraints:}&\;\eqref{boundary1},\;\eqref{boundary2}, \nonumber \\
\text{given:}&\; [\underline{t}_i^k, \overline{t}_i^k], v_i^d, \phi. \nonumber 
\end{align}

%\textcolor{red}{See really carefully the following result!}
\begin{theorem}[Solution existence]   \label{thm:feasibility}
%Let Algorithm 1 be executed in every time step along the trajectory of a CAV. Consider that if there exists a preceding HDV, it cannot decelerate faster than $-u_{\max}$.\label{as:hdv}
%Then for every CAV~$i$ and for all $t\in[t_i^0,t_i^\mathrm f]$ the admissible-control set
For a large value of $\kappa_s$ in the speed barrier conditions \eqref{speed_CBF_1}, \eqref{speed_CBF_2},
the set of control inputs,
\[
\mathcal{U}_i(t)\; :=\;
      \underbrace{[-u_{\max},u_{\max}]}_{\mathcal{U}_{\text{ctrl}}}
      \cap\
      \underbrace{\mathcal{U}_{\mathrm{speed}}}_{\text{\eqref{speed_CBF_1}--\eqref{speed_CBF_2}}}
      \cap\
      \underbrace{\mathcal{U}_{\mathrm{rear}}(t)}_{\text{\eqref{CBF_rear_end}}}
      \cap\
      \underbrace{\mathcal{U}_{\mathrm{cross}}(t)}_{\text{\eqref{boundary1},\eqref{boundary2}}}
\]
is a non-empty interval; consequently a solution to Problem 1, and the control action $u_i$, always exists.
\end{theorem}

\begin{proof}
We prove existence by computing the set intersections from left to right.
\paragraph{Control bounds}
By \eqref{eq:contro_bounds}, $\mathcal U_{\mathrm{ctrl}}=[-u_{\max},u_{\max}]$ is a non-empty interval by definition.

\paragraph{Speed limits}
For large values of $\kappa_s$, the speed barrier conditions \eqref{speed_CBF_1}–\eqref{speed_CBF_2} approximate a step function with the feasible space,
\[
\mathcal U_{\mathrm{ctrl}}\cap\mathcal U_{\mathrm{speed}}=
\begin{cases}
[-u_{\max},0], & v_i=v_{\max},\\[2pt]
[0,u_{\max}], & v_i=0,\\[2pt]
[-u_{\max},u_{\max}], & \text{else},
\end{cases}
\]
which is always non-empty because $0\in[-u_{\max},u_{\max}]$.

\paragraph{Rear-end safety}
The rear-end safety barrier condition \eqref{CBF_rear_end} imposes an upper bound $u_{rear} \geq -u_{\max}$ by design of \eqref{CBF_rear_end}.
When $v_i = 0$, by construction $u_{rear} \geq 0$ x, which implies,
\begin{equation*}
    \mathcal U_{\mathrm{ctrl}}\cap\mathcal U_{\mathrm{speed}}
      \cap\mathcal U_{\mathrm{rear}}(t)
      =[0,\,\min\{0,u_{\mathrm{rear}}\}]=\{0\}.
\end{equation*}
When $v_i \geq 0$,
\begin{equation}
    \mathcal U_{\mathrm{ctrl}}\cap\mathcal U_{\mathrm{speed}}      \cap\mathcal U_{\mathrm{rear}}(t)
      =[-u_{\max},\,\min\{0,u_{\mathrm{rear}}\}],
\end{equation}
which is a non-empty interval.

\paragraph{Crossing-time window} %\textcolor{red}{Not extrmely confident here if the extreme case is enough}
%We continue with the extreme case where the CAV follows the speed limit.
The crossing time barrier conditions \eqref{boundary1}, \eqref{boundary2} yield a non-empty interval
\begin{equation}
\mathcal{U}_{\mathrm{cross}}(t)=
      [\underline u_{\mathrm{cross}}(t),\overline u_{\mathrm{cross}}(t)].
\end{equation}
The overall feasible set can only become infeasible in two cases.
First, if
\begin{align} \label{eq:cond-1}
    \underline{u}_{\mathrm{cross}}(t) > \min\{0, u_{rear}(t)\},
\end{align}
we select a later crossing interval, which always exists under Assumption 4.
That is, we decrease $\Delta t_{i,1}$, and therefore $\underline{u}_{\mathrm{cross}}(t)$, until \eqref{eq:cond-1} no longer holds.
Second, if
\begin{align} \label{eq:cond-2}
    \overline{u}_{\mathrm{cross}}(t) < 
    \begin{cases}
        0 &\text{ if } v_i = 0, \\
        -u_{\max} &\text{ otherwise}.
    \end{cases}
\end{align}
then there is no feasible solution.
By Theorem \ref{thm:no-overshoot}, we switch braking strategies and set $u_{\mathrm{cross}}(t)$ to the rear-end stopping constraint when $\Delta t_{i,1}$ satisfies \eqref{thm:no-overshoot}.
This implies $u_{\mathrm{cross}}(t)\geq u_{\mathrm{rear}}(t)$, where equality is strict if there is no preceding vehicle.
By definition, this does not satisfy \eqref{eq:cond-2}, and thus a feasible action exists.
\end{proof}

\begin{algorithm}[ht]
\caption{CAV control algorithm}
\label{alg:cav_control}
\DontPrintSemicolon

Find the smallest $k$ such that $[\underline t_i^k,\overline t_i^k]$ satisfies \eqref{condition1}\;

\While{$p_i\le p_i^{\mathrm{tr}}\; \&\; t\geq t_i^0$}{
    Compute distance $\Delta p_i$\;
    Compute $\Delta t_{i,1}$ and $\Delta t_{i,2}$\;

    \If{condition \eqref{eq:cross-time-limit} is \emph{not} satisfied}{
        $k \leftarrow k+1$\;
        Update $[\underline t_i^k,\overline t_i^k]$ until condition~\eqref{eq:cross-time-limit} holds\;
    }

    \If{$\Delta t_{i,2} > \sqrt{2\Delta p_i/u_{\max}}$}{
        Replace \eqref{boundary2} with \eqref{CBF_rear_end} and substitute $\delta_i(t)=p_i^{\mathrm{tr}}$\;
    }

    \If{\textbf{Problem~1} is feasible}{
        Apply \textbf{Problem~1}\;
    }
    \Else{
        $k \leftarrow k+1$\;
        Update $[\underline t_i^k,\overline t_i^k]$
    }
}
    Apply \textbf{Problem~1} \emph{excluding} crossing-time constraints
        \eqref{boundary1}–\eqref{boundary2}\;
\end{algorithm}

\subsection{Human Driven Vehicles} \label{HDVs}
In order to incorporate human driver behavior in our scenario, we utilized a well-known car following model of the literature, the Intelligent Driver Model (IDM) model; see \cite{albeaik2022limitations}. The acceleration $u_i$ of a vehicle based on IDM is given by,
\begin{equation}
u_i = u_{\max} \left[ 1 - \left( \frac{v_i}{v_{\text{des}}} \right)^{\xi} - \left( \frac{s^*(v_i, \Delta v_i)}{s_i(t)} \right)^2 \right]
\end{equation}
where: $v_i$ is the current speed of the vehicle, $v_{\text{des}}$ is the desired speed, $\xi$ is a model parameter, $s_i(t)$ is the gap to the leading vehicle, $\Delta v_i = v_i - v_{\text{lead}}(t)$ is the speed difference between the vehicle and the vehicle ahead. The desired minimum gap \( s^*(v_i, \Delta v) \) is given by: $s^*(v_i, \Delta v_i) = \gamma_i + v T + \frac{v_i \Delta v_i}{2 \sqrt{u_{\max} \beta}}$, $T$ is the desired time headway, and $\beta$ is the comfortable deceleration.\

% \begin{problem} \label{prb:reactive}
% At each time, CAV $i$ applies the control action that satisfies,
%     \begin{align*}
%     \min_{u_i}\,&  (u_i - \alpha(v_i^d - v_i))^2 \\
%     \text{subject to:}\\
%     \text{crossing constraints:}\, &\eqref{eq:cbf-first}, \eqref{eq:cbf-second}, \\
%     \text{rear-end safety:}\, &\eqref{eq:cbf-rear-end}, \\
%      \text{control bounds:}\,&|u| \leq u_{\max}, \\
%      \text{given:}\,&(p_1, t_1, t_2),
% \end{align*}
% where $[t_1, t_2]$ is the scheduled time interval to pass point $p_1$, and $\kappa$ is a gain for the speed bound CBF.
% \end{problem}

% \begin{theorem}[[X, Theorem 1]]  \label{thm:sln}
%     The solution to Problem \ref{prb:reactive} is,
%     \begin{equation}
%     \begin{aligned}
%         u_i^* &= \clamp(\alpha(v_i^d - v_i), \underline{u}, \overline{u})\\
%         &= \min(\max(\alpha(v_i^d - v_i),\, \underline{u}),\, \overline{u}),
%     \end{aligned}
%     \end{equation}
%     where
%     \begin{align}
%         \overline{u} &= \max\big(\min(\,\eqref{eq:cbf-second},\, \eqref{eq:cbf-rear-end}\,), -u_{\max}\big) \label{eq:u-lb}, \\
%         \underline{u} &= \min\big(\eqref{eq:cbf-first}, u_{\max}\big).  \end{align}
% \end{theorem}

\section{Handling right turns at red traffic lights}\label{Handling right turns}

In our framework, we also account for situations where vehicles are allowed to turn right at a red traffic signal. This scenario effectively becomes a merging problem, with the key requirement that vehicles must first stop at the red light to check for oncoming traffic before proceeding.

For HDVs, this behavior can be readily modeled using built-in procedures from traffic simulation tools such as SUMO or VISSIM, which emulate driver behavior at stop signs. In contrast, for CAVs, we design a dedicated merging strategy.

\begin{remark}
While our previous work \cite{le2024stochastic} addressed the problem of merging in mixed traffic environments, it required the assumption that CAVs had complete awareness of the state of other HDVs and CAVs in the network. In contrast, we consider a more restrictive and realistic scenario: CAVs can only communicate with traffic lights and perceive the state of surrounding vehicles using onboard sensors.
\end{remark}

Our control strategy for CAVs in this context is built on a simple but effective principle. We estimate two quantities:
\begin{itemize}
    \item The time required for the CAV to reach the \textit{conflict point} (i.e., the point where the main and merging roads meet),
    \item The time required for the nearest oncoming vehicle (if one exists) on the main road to reach the same conflict point.
\end{itemize}
By comparing these two time headways, we determine whether it is safe for the CAV to merge.

Consider the setup illustrated in Fig.~\ref{fig:merging}, where vehicle $i$ travels on the main road and vehicle $j$ is on the secondary road intending to merge. Suppose vehicle $j$ is a CAV. Let $d_j$ denote the distance between the traffic signal and the conflict point,
%point at which the CAV assesses the traffic flow and the conflict point, 
and let $d_i$ denote the distance between vehicle $i$ and the conflict point.

At the moment when CAV $j$ is at rest and checks the upcoming flow, we define $\tau_j$ as the time required for CAV $j$ to reach the merging point while accelerating at a rate of $\phi(v_j - v_j^d)$, as discussed in \eqref{unconstrained trajectry}. However, this acceleration is not constant, and that means we need to solve an ODE to identify the time $\tau_j$ that CAV $j$ will reach the merging point, that is, to cover the distance $d_j$. 

Specifically, since the acceleration follows a feedback law of the form $\dot{v}_j = \phi (v_j^d - v_j)$, we solve this first-order linear ODE with the initial condition $v_j(0) = 0$. The resulting velocity profile is:
\begin{equation}
v_j(t) = v_j^d (1 - e^{-\phi t}).
\end{equation}
By integrating this expression, we obtain the position of CAV $j$ over time:
\begin{equation}
p_j(t) = \int_0^t v_j(\tau) \, d\tau = v_j^d \left[t + \frac{1}{\phi}(e^{-\phi t} - 1)\right].
\end{equation}
Then, to determine the required time $\tau_j$, we solve the following equation,
\begin{equation} \label{ODE_upcoming}
d_j = v_j^d \left[\tau_j + \frac{1}{\phi}(e^{-\phi \tau_j} - 1)\right],
\end{equation}
where the solution is the Lambert $\mathit{W}$ function \cite{corless1996lambert}.
Additionally, 
%Since this equation is transcendental in $\tau_j$, a closed-form solution is not available, and
a numerical root-finding method (e.g., Newton--Raphson) may be used to determine $\tau_j$. Yet, this can be readily done in milliseconds using a commercial software like MATLAB.

\begin{remark}
Note that \eqref{ODE_upcoming} reflects the scenario in which the CAV has come to a complete stop to check the oncoming traffic. However, it would be possible to relax this assumption and consider that the CAV makes its decision earlier while still in motion. In that case, in equation \eqref{ODE_upcoming}, we must also consider the current speed speed of the CAV.
\end{remark}

Next, we estimate the time headway $\tau_i$ of vehicle $i$ on the main road. We assume that the CAV’s sensors can accurately estimate the position and velocity of vehicle $i$. To ensure robustness in the decision-making process, we consider the extreme scenario in which vehicle $i$ accelerates with its maximum allowable acceleration. The resulting time headway $\tau_i$ is characterized in the following proposition.

\begin{proposition}
Consider vehicle $i$ at time $t_i^c$ that constantly accelerates. Then, the fastest time headway $\tau_i$ is given by:
\begin{equation}
    \tau_i =
\begin{cases}
\displaystyle \frac{-v_i^c + \sqrt{(v_i^c)^2 + 2u_{\max} d_i}}{u_{\max}},& \text{if } d_i \le \dfrac{v_{\max}^2 - (v_i^c)^2}{2u_{\max}}, \\[1em]
\displaystyle \frac{v_{\max} - v_i^c}{u_{\max}} + \frac{d_i - \frac{v_{\max}^2 - (v_i^c)^2}{2u_{\max}}}{v_{\max}},& \text{if } d_i > \dfrac{v_{\max}^2 - (v_i^c)^2}{2u_{\max}},
\end{cases}
\end{equation}
where $v_i^c$ is the velocity at time $t_i^c$, that is, the time where the vehicle from the secondary road checks the upcoming flow.
\end{proposition}

\begin{proof}
We analyze the motion of vehicle $i$ with current speed $v_i^c$, under constant acceleration $u_{\max}$, and subject to the speed limit $v_{\max}$. Note that $\dfrac{v_{\max}^2 - (v_i^c)^2}{2u_{\max}}$ is the distance the vehicle $i$ covers with maximum acceleration until it reaches the maximum speed $v_{\max}$. Let $d_i$ be the distance to be covered. We distinguish two cases. \textbf{Case 1: $d_i \le \dfrac{v_{\max}^2 - (v_i^c)^2}{2u_{\max}}$.} In this case the vehicle does not reach the speed limit before covering the distance $d_i$. Thus, the motion is governed by the kinematic equation $d_i = v_i^c\,\tau_i + \frac{1}{2}u_{\max}\,\tau_i^2$.
By the quadratic formula, the positive solution is $\tau_i = \frac{-v_i^c + \sqrt{(v_i^c)^2 + 2u_{\max}d_i}}{u_{\max}}.$ \textbf{Case 2: $d_i > \dfrac{v_{\max}^2 - (v_i^c)^2}{2u_{\max}}$.} Here the vehicle reaches the speed limit $v_{\max}$ before covering the entire distance $d_i$. Let $d_a = \frac{v_{\max}^2 - (v_i^c)^2}{2u_{\max}}$ be the distance covered during the acceleration phase from $v_i^c$ to $v_{\max}$. The time required for this phase is $t_a = \frac{v_{\max} - v_i^c}{u_{\max}}.$ After reaching $v_{\max}$, the vehicle covers the remaining distance $d_i - d_a$ at constant speed $v_{\max}$. The time for the cruising phase is
$t_c = \frac{d_i - d_a}{v_{\max}}.$
Thus, the total time headway is $
\tau_i = t_a + t_c = \frac{v_{\max} - v_i^c}{u_{\max}} + \frac{d_i - \frac{v_{\max}^2 - (v_i^c)^2}{2u_{\max}}}{v_{\max}}$. That completes the proof.
\end{proof}

Based on these two time headways, we consider the condition
\begin{equation} \label{merging_condition}
\tau_j \leq \tau_i - \tau_s
\end{equation}
which, when satisfied, allows the CAV on the secondary road to merge onto the main road, assuming that the maneuver poses no risk of collision. Here, $\tau_s \geq 0$ is an additional safety margin chosen by the designer.

\begin{figure}
    \centering
    \includegraphics[width=0.7\linewidth]{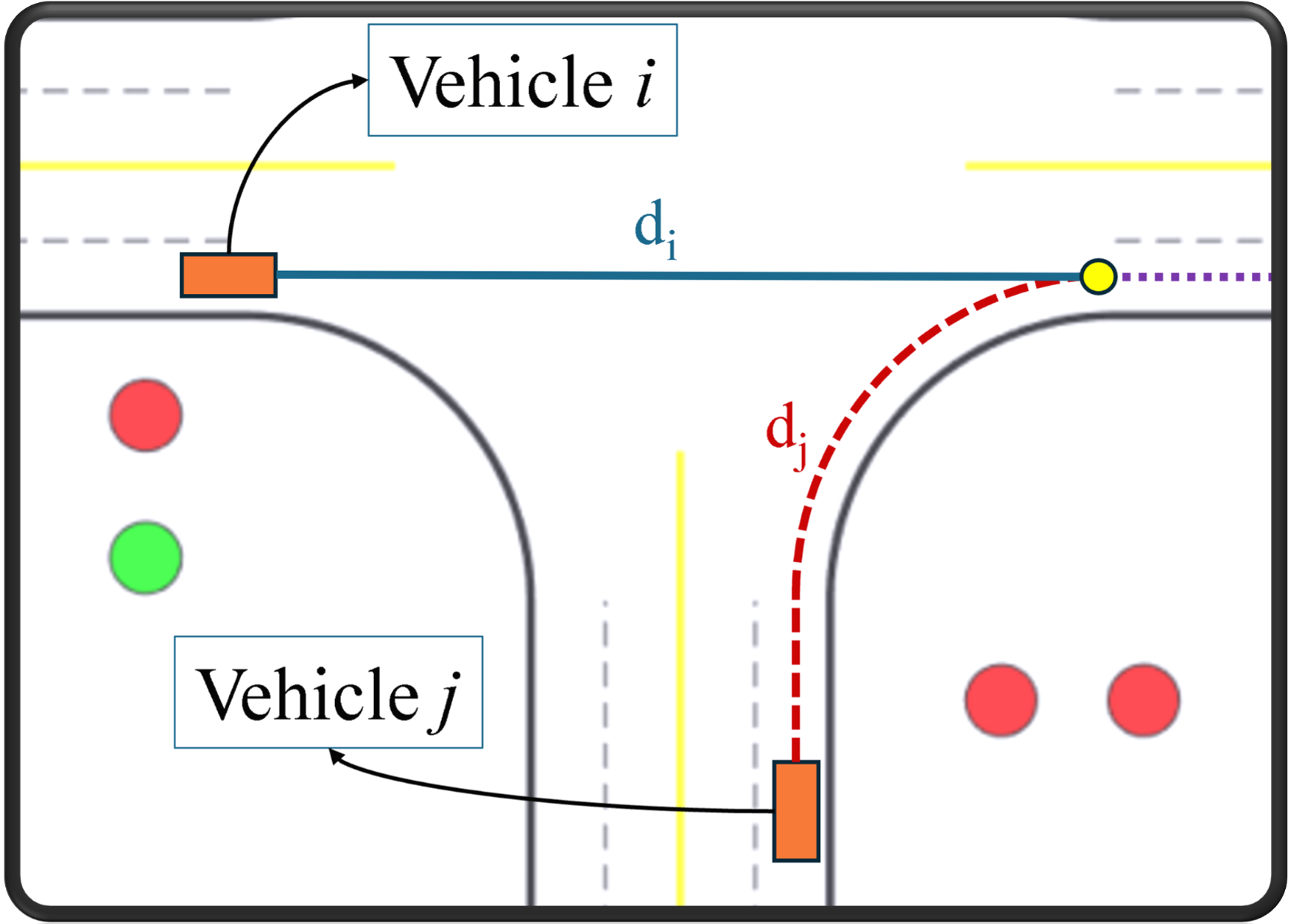}
    \caption{An example of a vehicle turning right,  checking the upcoming flow on red.}
    \label{fig:merging}
\end{figure}

\section{Simulation Results} \label{simulations}
In this section, we conduct numerical simulations using MATLAB. We consider a traffic scenario, shown in Fig. \ref{fig:Intersection}, where a traffic flow of 5000 vehicles per hour arrives at the intersection from all directions. The traffic light region is set to a range of 200 meters. We also set the following values for our parameters, $v_{\max}=22$m/s, $v_{\text{des}}=12$m/s, $u_{\max}=5$m/$\text{s}^2$, $\Delta t=0.05$ and $\tau_s=1.5$s.  We first examine in detail the performance and behavior of our controller under specific traffic conditions with the selection of the gains as $\phi=0.25$ $\kappa_T=0.04$, $\kappa_R=0.2$ and $\kappa_{\text{imag}}=0.05$ where $\kappa_{\text{imag}}$ denotes the gain imposed by the rear-end constraint \eqref{CBF_rear_end} when condition in Fig. $\ref{fig:switching}$ is satisfied, implying the imaginary vehicle at the stop line of the intersection. We denote it by $\kappa_{\text{imag}}$ to distinguish it from the gain $\kappa_R$, which is associated with the rear-end constraint between two vehicles. Finally we also provide quantitative results under varying penetration rates of CAVs.

In Fig. \ref{fig:path_5_positions}, we see the trajectories of five vehicles (two HDVs and three CAVs) over a horizon of 80 seconds. The y-axis represents the position of each vehicle, starting from the moment it enters the traffic light region. The first vehicle is a CAV. When the vehicles enter the region, the traffic light is green and remains green until $t = 15$ s.

Let us now examine the control behavior of the first CAV. In Fig. \ref{fig:CAV_3_control}, we observe that the initial control input of the vehicle follows the reference input given by \eqref{unconstrained trajectry}. Until $t = 7.4$ s (see box 1), the constraints on the control input and the rest of the system do not interfere with the reference input. At that point, however, constraint \eqref{boundary2} requires the vehicle to accelerate in order to guarantee it will cross the intersection before the traffic light changes (see box 5). This acceleration causes a deviation from the desired speed, resulting in a gradual decrease of the reference control input, as seen in the figure.

Once the vehicle ensures it can safely cross the traffic light, it decelerates and realigns with the reference control input at $t = 14.5$ s (see box 2). Shortly after, while crossing the intersection, the vehicle detects a preceding vehicle from a different path at $t = 16$ s (box 3). It must then satisfy the rear-end constraint and decelerates accordingly, again deviating from the reference input. Finally, at $t = 22.5$ s, it returns to the reference input and resumes the desired speed. The corresponding speed profile is shown in Fig. \ref{fig:path_5_Speeds}, where this vehicle is the first to enter the traffic light region. The figure also includes the speed trajectories of the remaining vehicles.

We now focus on the control behavior of the second vehicle, which is again a CAV. From the very beginning, the CAV deviates from the reference trajectory (see box 1), as it must obey the red traffic light until it becomes green (see box 4). Notably, the condition described in \eqref{eq:cross-time-limit} is not satisfied and as a result, the vehicle is required to stop at the traffic light using the rear end constraint \eqref{CBF_rear_end} by considering an imaginary vehicle at the stop line of the intersection, as explained in Fig.~\ref{fig:switching}.

Shortly before the signal turns green (see box 2), the condition in \eqref{eq:cross-time-limit} is eventually satisfied, allowing the vehicle to accelerate and reach the bound (21). Finally, at $t = 59$ s, the vehicle aligns with the reference trajectory, which in this case does not conflict with any rear-end constraints. As depicted in Fig. \ref{fig:path_5_Speeds}, the control input aim to drive the speeds of the vehicles to the desired speed of 12 m/s.

\subsection{Sensitivity analysis of the gains}
Note that the selection of the gains in the barrier conditions significantly influences the controller’s performance. Smaller gain values tend to make the controller more proactive, allowing it to adapt to the corresponding constraint earlier. In contrast, larger gain values result in a more reactive behavior, where the controller enforces the constraint more abruptly, just before it becomes active.

To illustrate this, consider the case of a CAV approaching an intersection with a red traffic light. In Fig. \ref{fig:k_imag_analysis_position}, we show the trajectory of a CAV starting from the same initial position under varying values of the gain $\kappa_{\text{imag}}$. Recall that we define the gain $\kappa_{\text{imag}}$ as the gain associated with the traffic-light constraint, which considers an imaginary vehicle positioned at the stop line when condition in \eqref{thm:no-overshoot} is not satisfied. As seen in the figure, when $\kappa_{\text{imag}} = 0.01$, the vehicle decelerates earlier, whereas with $\kappa_{\text{imag}} = 0.1$, the deceleration occurs closer to the traffic light. This behavior highlights how smaller gain values enable earlier adaptation to safety constraints, while larger gains delay the reaction.

A similar pattern emerges when examining the effect of the gain $\kappa_T$, associated with CBF constraint \eqref{boundary1}, which enforces a lower-bound constraint to ensure timely crossing of the intersection. In Fig. \ref{fig:k_t_analysis_control}, we observe how different values of $\kappa_T$ influence the control bounds. For $\kappa_T = 0.01$, the upper bound on acceleration starts near 2 m/s², encouraging the vehicle to accelerate sooner. In contrast, for $\kappa_T = 0.3$, the upper bound begins at a lower value and increases more gradually as the vehicle approaches the intersection. This again confirms that smaller gain values result in earlier constraint enforcement. This behavior is further supported in Fig. \ref{fig:k_t_analysis_position}, where the vehicle crosses the intersection earlier for $\kappa_T = 0.01$ compared to $\kappa_T = 0.05$ and $\kappa_T = 0.1$.

Finally, Fig. \ref{fig:phi_analysis} illustrates how the gain $\phi$ in the unconstrained trajectory \eqref{unconstrained trajectry} affects the vehicle's speed profile when starting from rest at a red light. Since this gain scales the control action, higher values of $\phi$ lead to more abrupt acceleration, while smaller values result in smoother transitions. This is clearly depicted in Fig. \ref{fig:phi_analysis}, where for $\phi = 0.05$, convergence to the desired speed occurs more gradually, whereas for $\phi = 0.55$, the vehicle reaches the desired speed more quickly. 

As already discussed, the CBF structure defines lower and upper bounds on the control input. Thus, the selection of the gains can have direct implications on fuel consumption, as aggressive acceleration patterns tend to increase fuel usage, whereas smoother trajectories are more energy-efficient. Similarly, the gain $\phi$, which scales the control effort toward the reference (unconstrained) trajectory, affects how aggressively a vehicle attempts to match its desired speed. Higher values of $\phi$ result in faster convergence but at the cost of increased acceleration effort. Consequently, by carefully tuning the CBF gains and $\phi$, one can prioritize different performance objectives—minimizing travel time, enhancing safety, or improving fuel efficiency—depending on the needs of the application.

\subsection{Turn on Red Light}

We now examine the behavior of a CAV when a right turn on red is permitted. In Fig. \ref{fig:turn_on_green_both} (left), we present two scenarios involving an HDV approaching and crossing a traffic light. In Case 2, the HDV enters the traffic light region later than in Case 1. Fig. \ref{fig:turn_on_green_both} (right) shows the trajectory of a corresponding CAV that intends to make a right turn and has a potential conflict with the HDV, in both cases. 

In the first case, the CAV remains stopped at the red light for a longer period because the HDV from the conflicting direction approaches in a manner that violates the condition discussed in \eqref{merging_condition}. In contrast, in Case 2, since the HDV enters the traffic light region later, the CAV is able to safely proceed with the right turn on red. This is because the condition for conflict avoidance is satisfied, and no danger is imposed on either vehicle.

\subsection{Quantitative Analysis of Gain Parameters}

To evaluate the impact of the gain parameters on overall traffic flow, we conducted a quantitative analysis by measuring the average time vehicles spend within the traffic light region under identical initial conditions and varying gain values.

Table~\ref{tab:different gains} presents the average dwell times corresponding to different values of the gains $\kappa_R$, $\kappa_T$, and $\kappa_{\text{imag}}$, with $\phi$ fixed at 0.25. We observe that higher gain values ($\kappa_R=0.06$, $\kappa_T=0.3$, $\kappa_{\text{imag}}=0.075$) result in a lower average dwell time of 55.758 seconds. In contrast, for lower gain values ($\kappa_R=0.03$, $\kappa_T=0.1$, $\kappa_{\text{imag}}=0.02$), the average dwell time increases to 57.708 seconds. This suggests that higher gains promote shorter travel times, likely because vehicles respond to constraints later, avoiding early deceleration—particularly in the context of rear-end safety constraints.

Similarly, in Table~\ref{tab:different phis}, we report the average dwell times for fixed gain values ($\kappa_R=0.06$, $\kappa_T=0.3$, $\kappa_{\text{imag}}=0.075$) and varying values of $\phi$. The results show that increasing $\phi$ leads to reduced travel times, which is expected since higher values of $\phi$ produce stronger control efforts that accelerate convergence to the reference speed.

Finally, Table~\ref{tab:different penetrations} shows how the average travel time changes with different CAV penetration rates. As expected, higher CAV penetration results in shorter travel times. It is important to note that in this study, CAVs are assumed to communicate only with the infrastructure (i.e., traffic lights) while they perceive their immediate surroundings using onboard sensors, cameras, or lidar. Enabling inter-vehicle communication introduces additional hardware and software requirements and opens up alternative modeling and control approaches; see our previous work \cite{tzortzoglou2024feasibility} for further discussion.

\begin{figure*}
    \centering
    \includegraphics[width=0.92\linewidth]{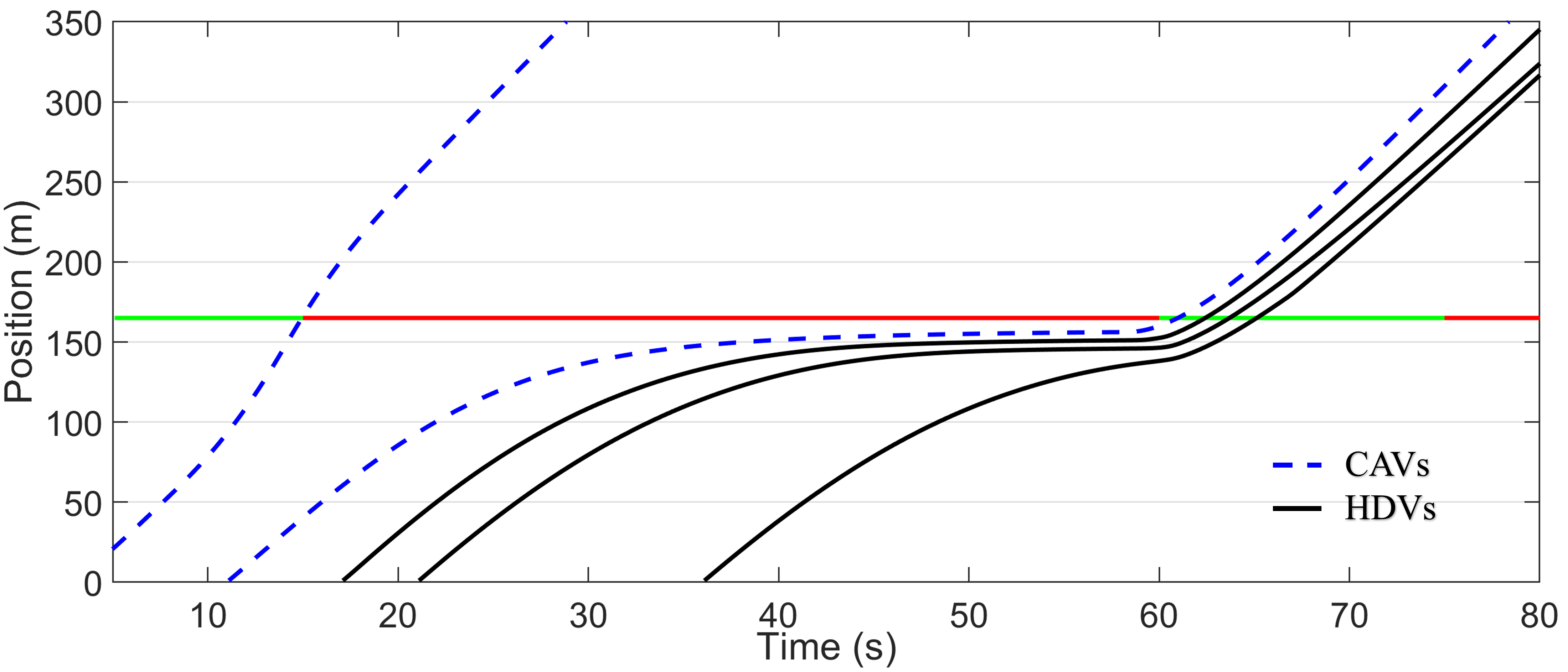}
    \caption{Position trajectories.}
    \label{fig:path_5_positions}
\end{figure*}
\begin{figure*}
    \centering
    \includegraphics[width=0.9\linewidth]{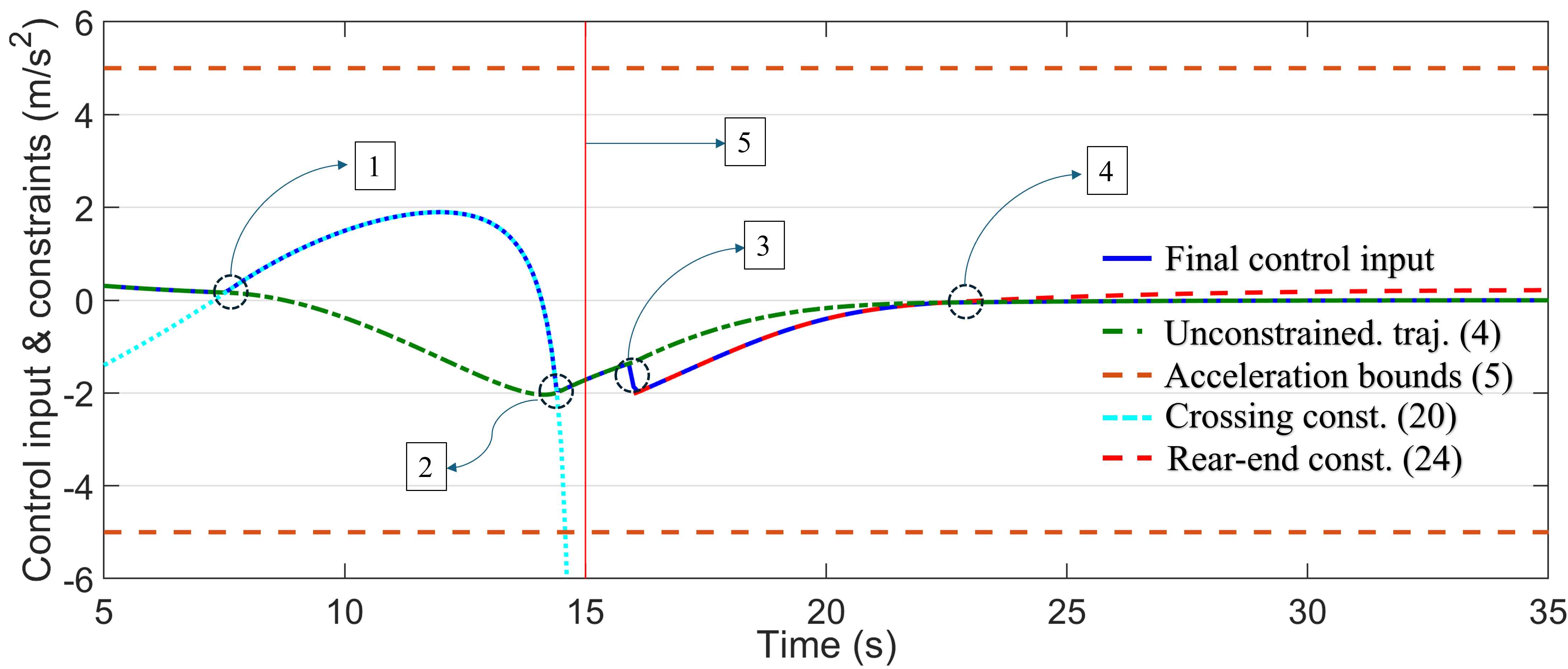}
    \caption{Control input and constraints of the first CAV.}
    \label{fig:CAV_3_control}
\end{figure*}
\begin{figure*}
    \centering
    \includegraphics[width=0.9\linewidth]{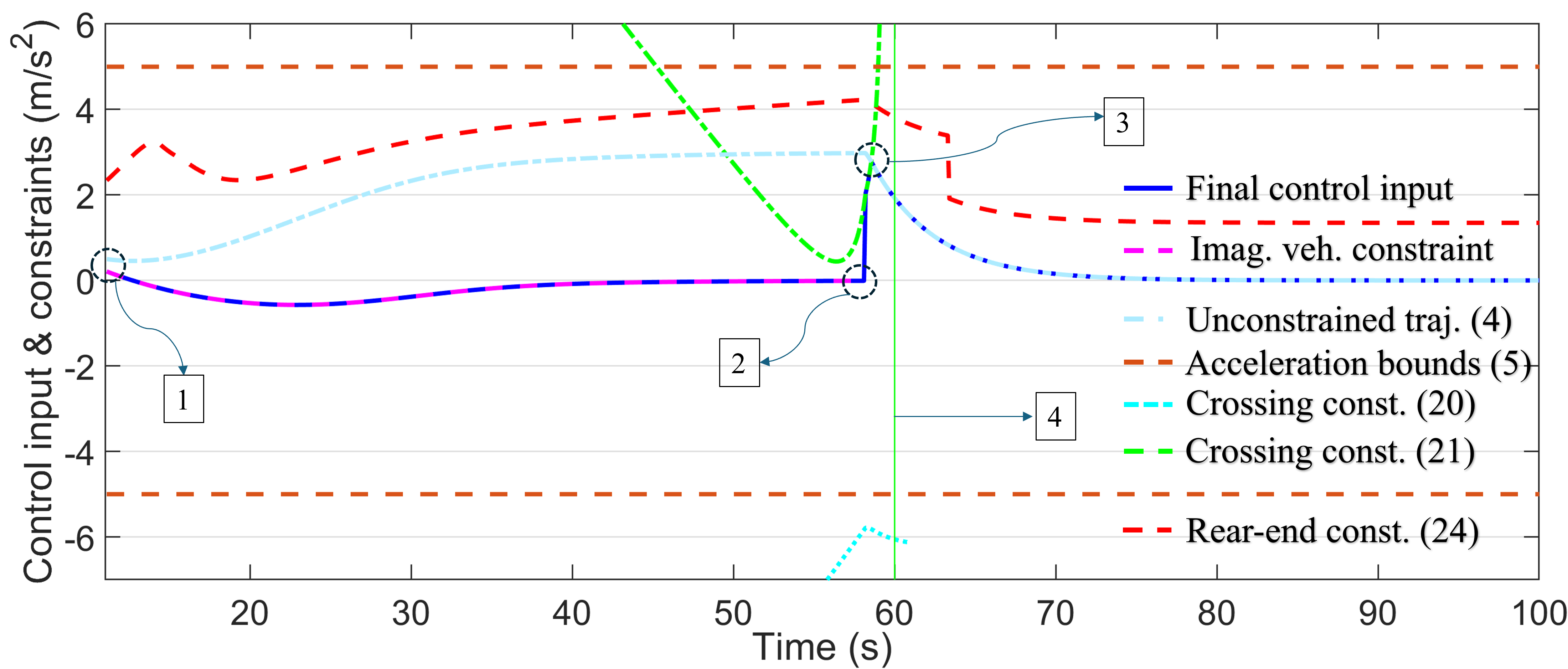}
    \caption{Control input and constraints of the second CAV.}
    \label{fig:CAV_11_control}
\end{figure*}
\begin{figure*}
    \centering
    \includegraphics[width=0.91\linewidth]{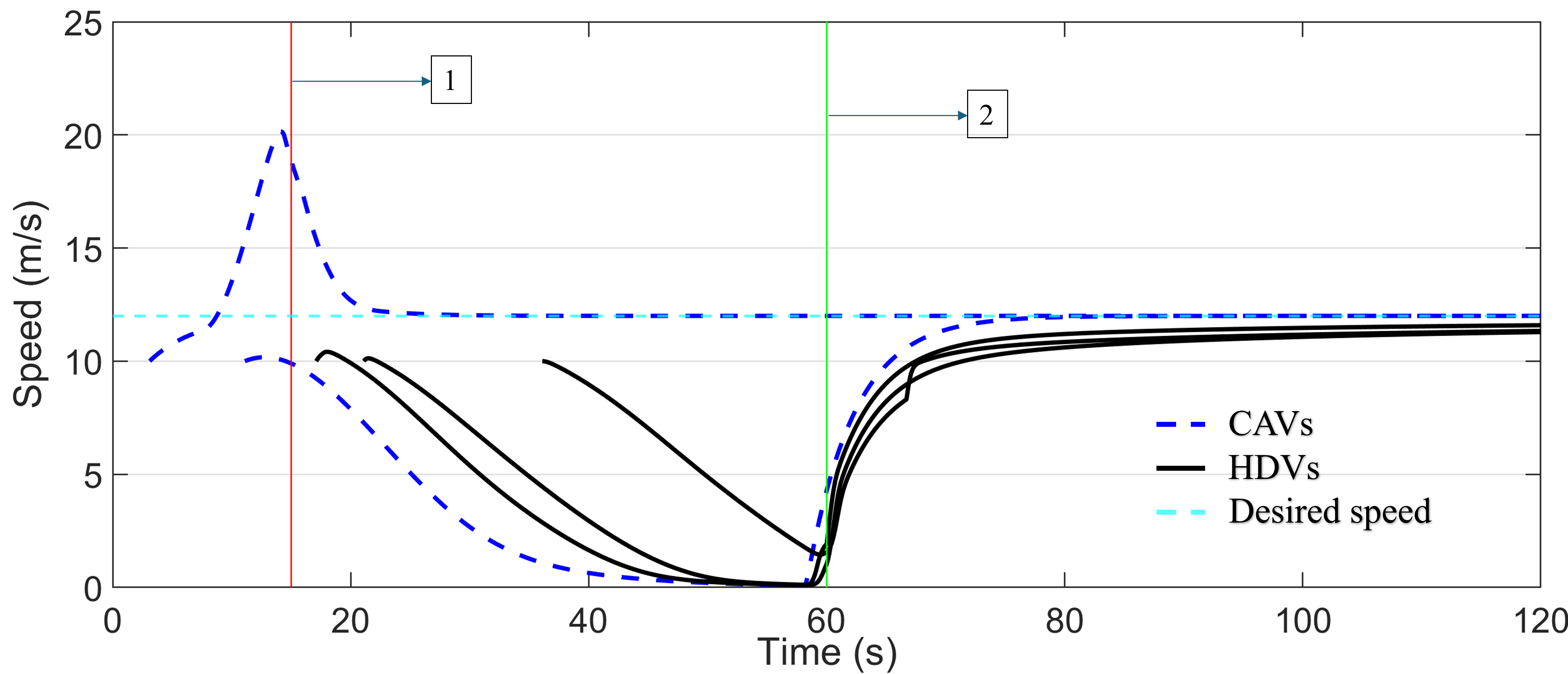}
    \caption{Speed trajectories.}
    \label{fig:path_5_Speeds}
\end{figure*}
\begin{figure}
    \centering
    \includegraphics[width=0.9\linewidth]{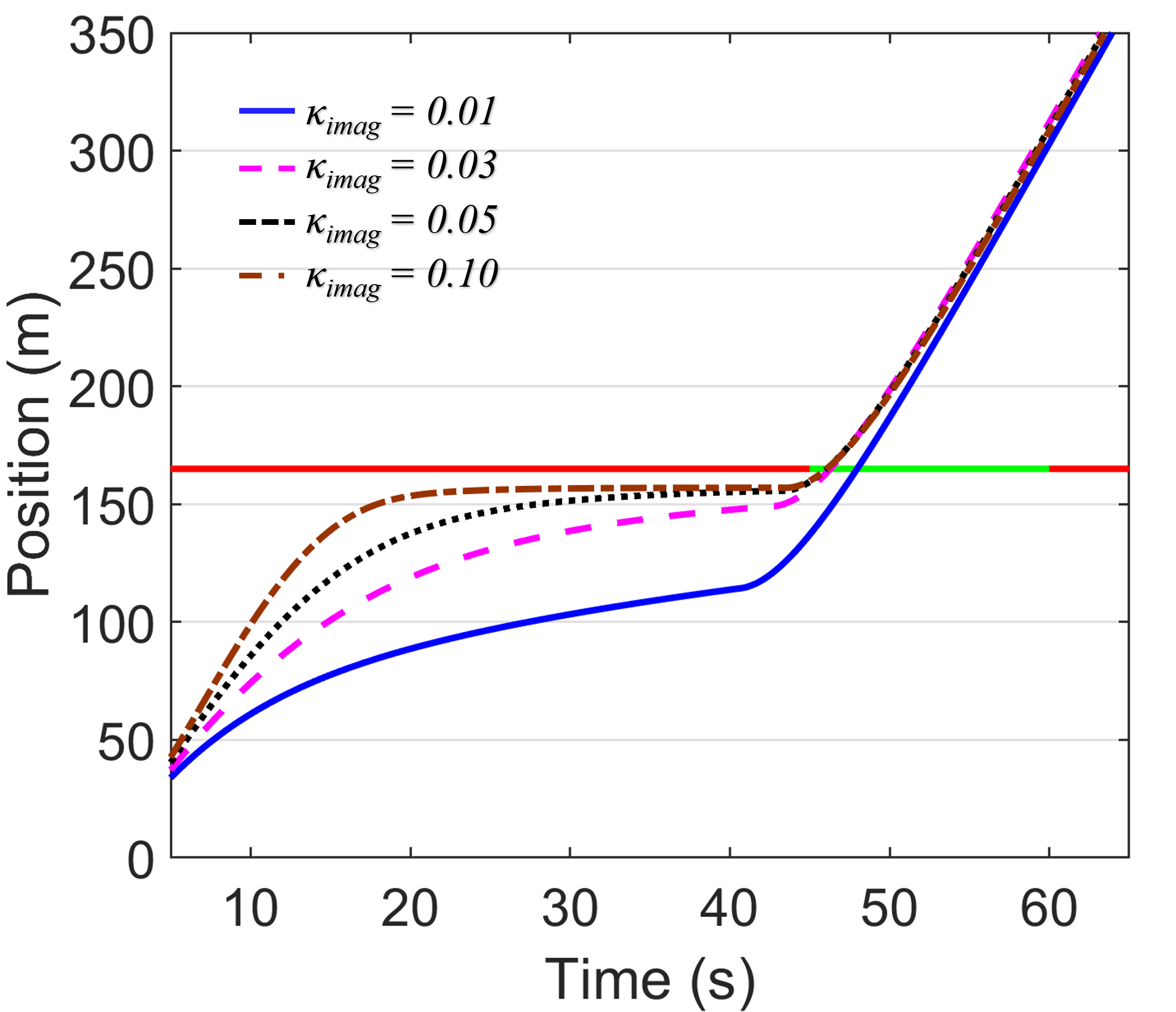}
    \caption{Different position trajectories for different gains associated (24).}
    \label{fig:k_imag_analysis_position}
\end{figure}
\begin{figure}
    \centering
    \includegraphics[width=0.9\linewidth]{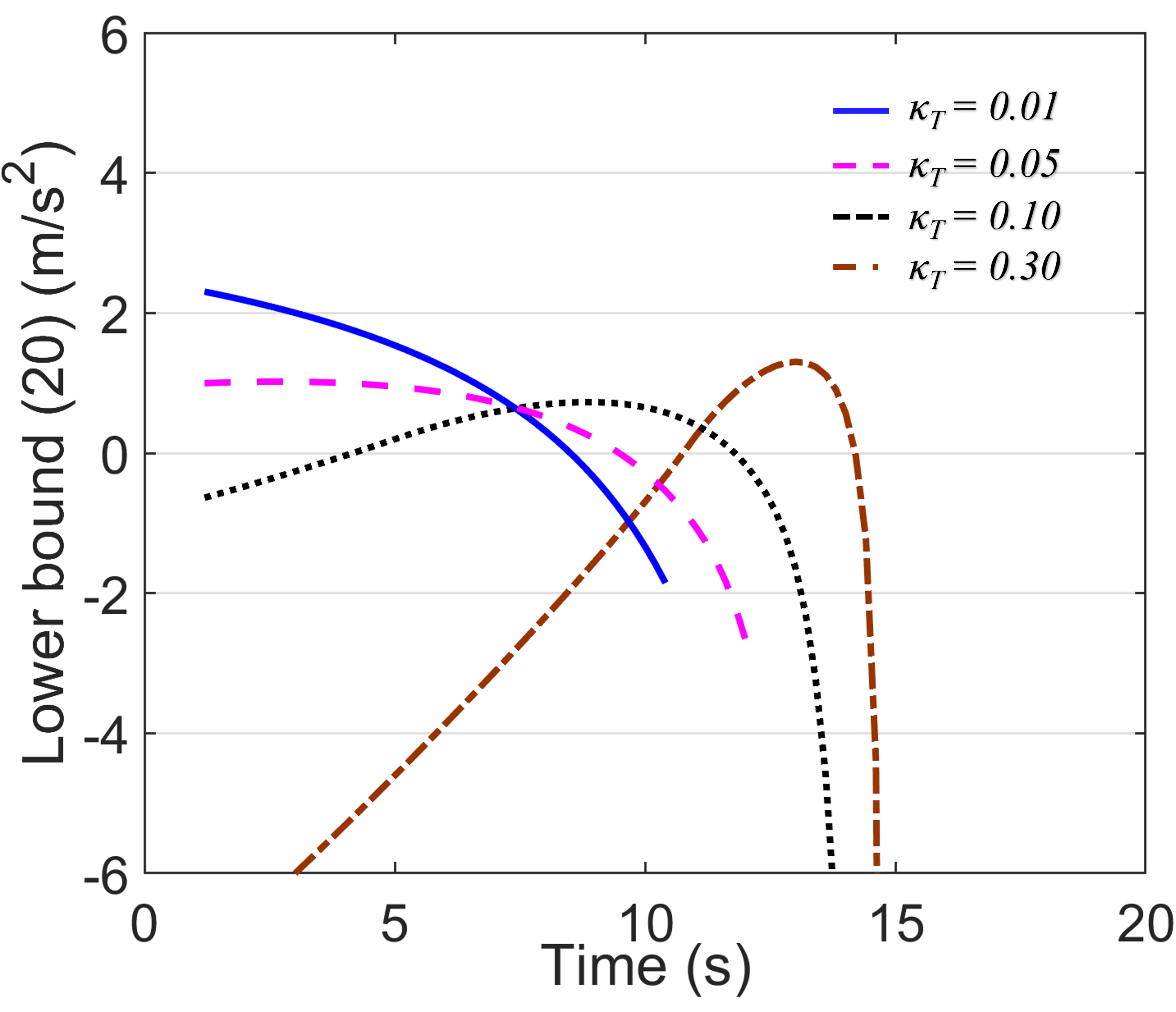}
    \caption{Different control bounds for different gains associated with (20).}
    \label{fig:k_t_analysis_control}
\end{figure}
\begin{figure}
    \centering
    \includegraphics[width=0.9\linewidth]{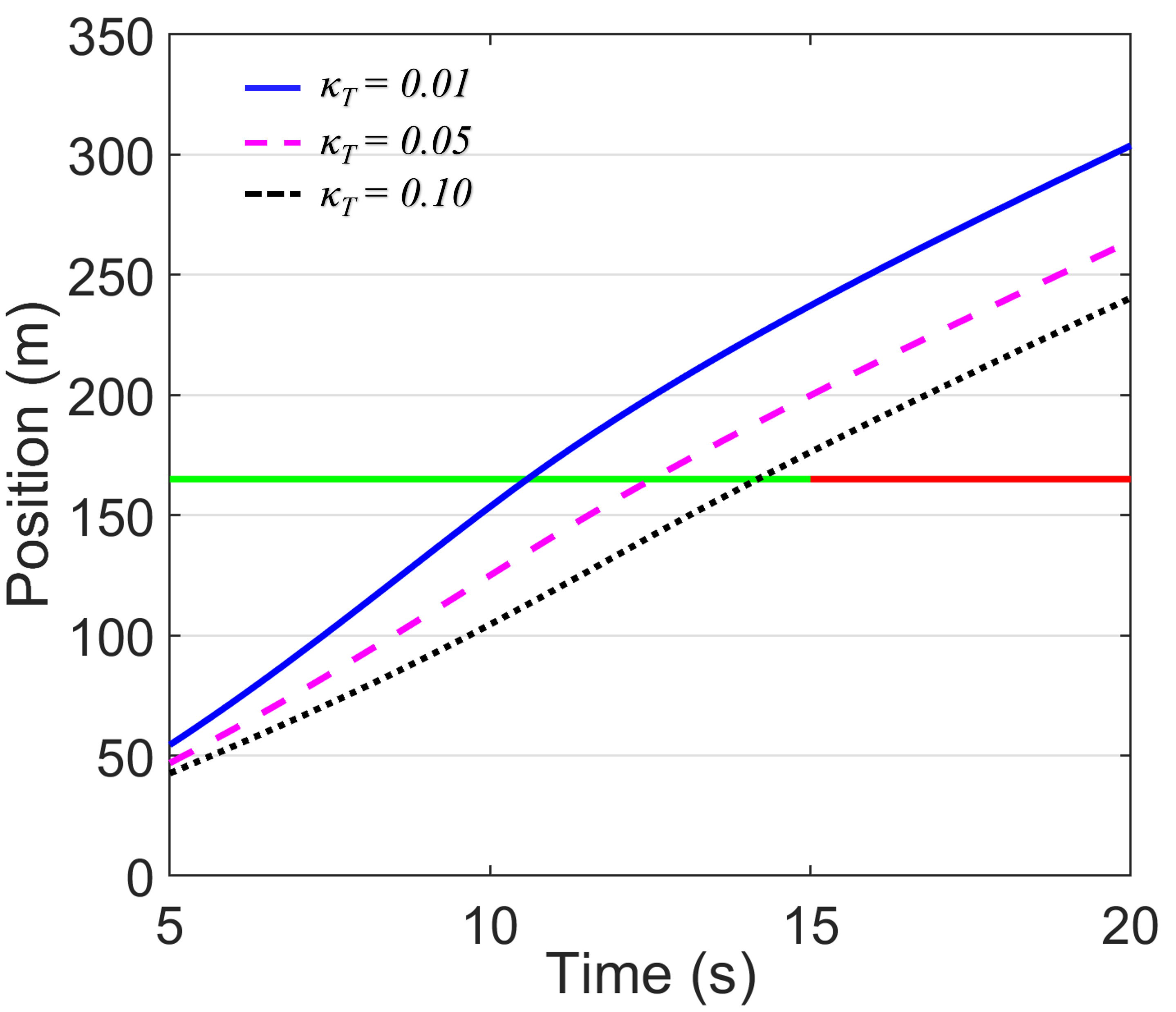}
    \caption{Different position trajectories for different gains associated with (20).}
    \label{fig:k_t_analysis_position}
\end{figure}
\begin{figure}
    \centering
    \includegraphics[width=0.9\linewidth]{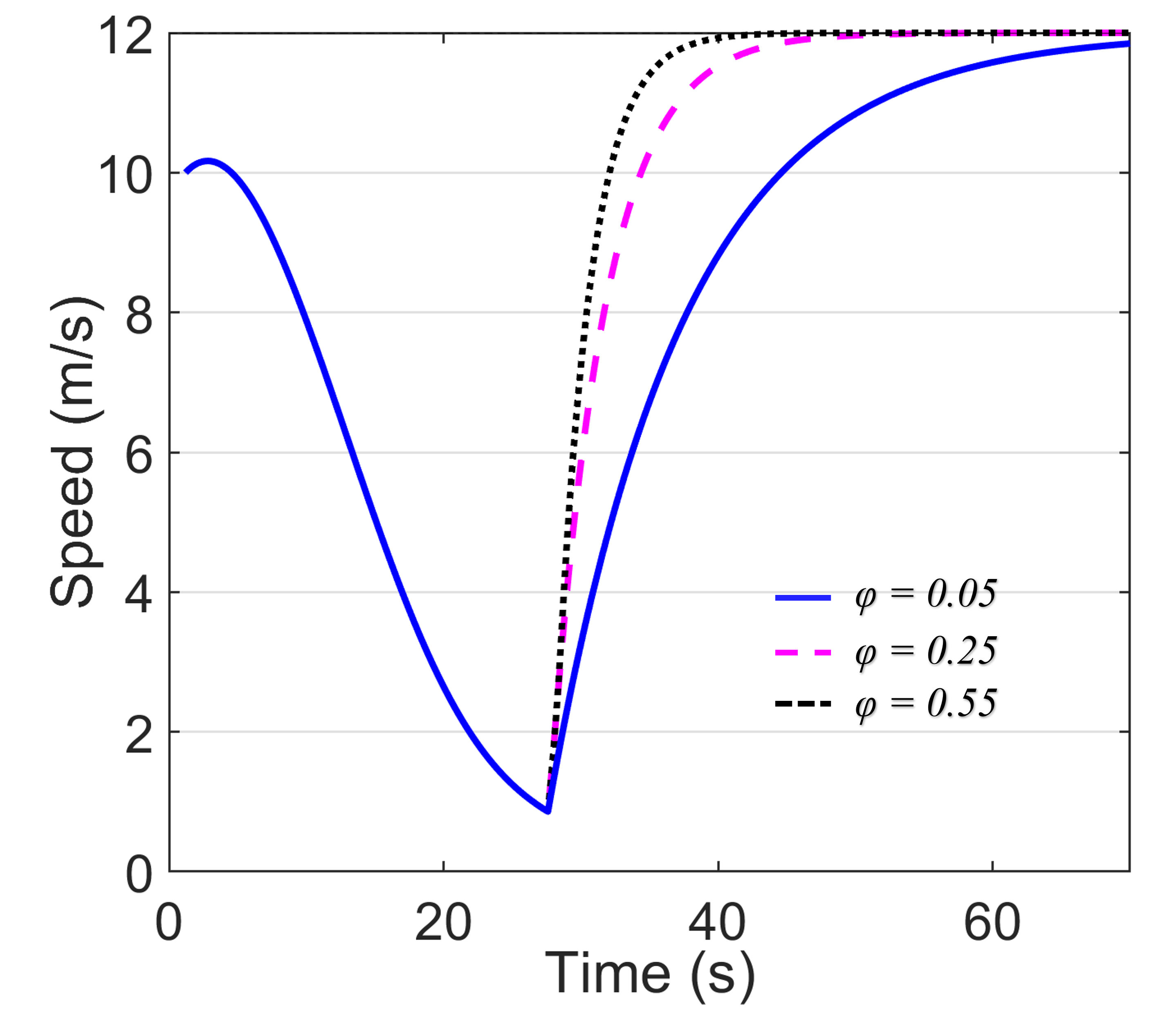}
    \caption{Different speed convergence for different values of $\phi$.}
    \label{fig:phi_analysis}
\end{figure}

\begin{figure*}[htbp]
   % \centering
%    \begin{subcaptiongroup}[b]{0.45\linewidth}
        \centering
  \begin{tabular}{p{0.45\textwidth}}
        \includegraphics[width=\linewidth]{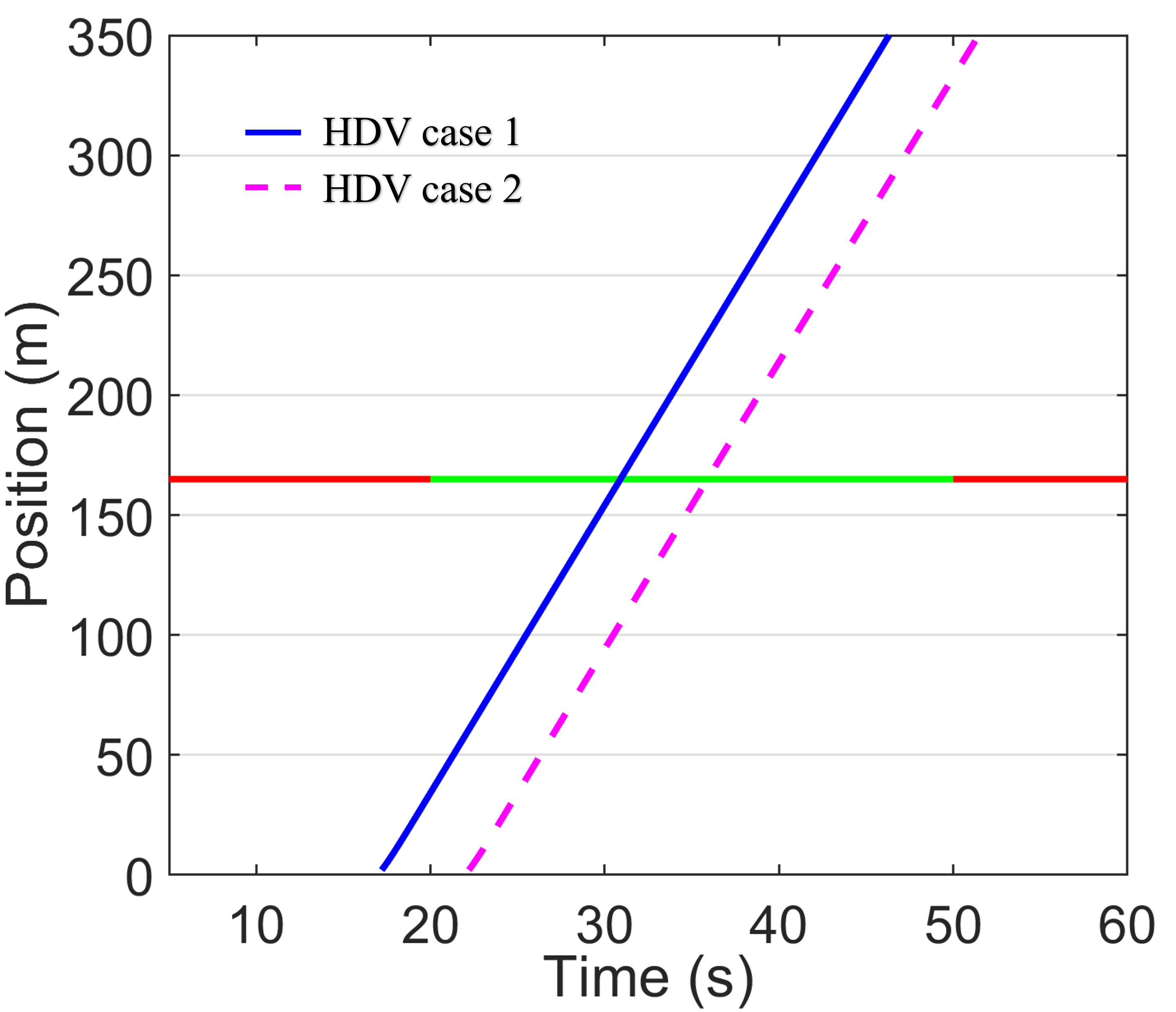} \\
        \centering
            \small{An HDV entering and crossing the intersection at two different times.}
            \end{tabular}
%        \caption{An HDV entering and crossing the intersection at two different times.}
%        \label{fig:turn_on_green_on_hdvs}
%    \end{subcaptiongroup}
    \hfill
%    \begin{subcaptiongroup}[b]{0.45\linewidth}
  \begin{tabular}{p{0.45\textwidth}}
        \includegraphics[width=\linewidth]{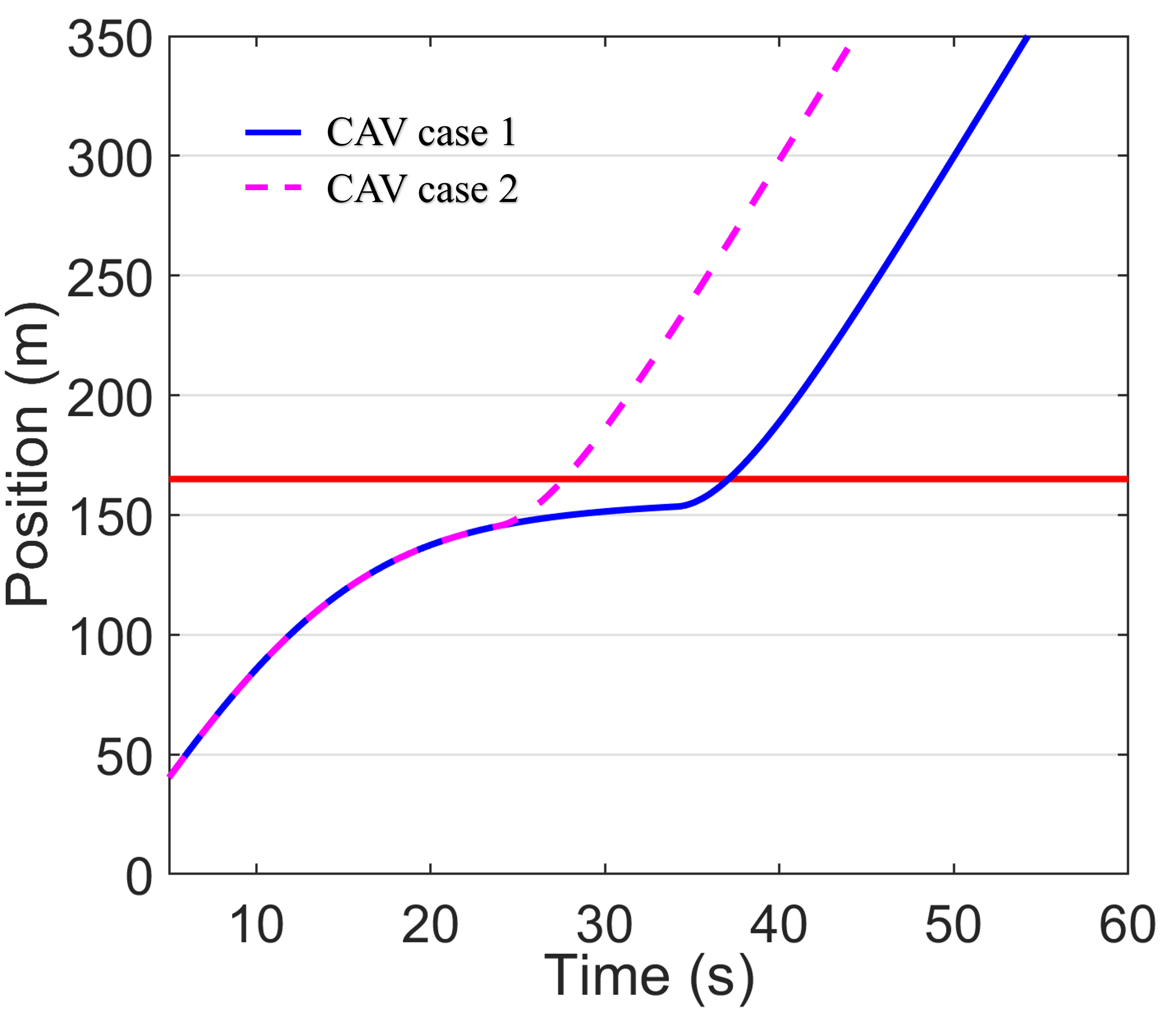} \\
        \centering
            \small{A CAV entering and crossing the intersection, while sharing a conflicting point with the HDV.}% of Fig.~\ref{fig:turn_on_green_on_hdvs}}
 %       \caption{A CAV entering and crossing the intersection, while sharing a conflicting point with the HDV of Fig.~\ref{fig:turn_on_green_on_hdvs}.}
%        \label{fig:turn_on_green_cavs}
%    \end{subcaptiongroup}
  \end{tabular}
    \caption{Illustration of two different cases of HDV and CAV behaviors when CAV has to turn on red. }
    \label{fig:turn_on_green_both}
\end{figure*}

\begin{table}[ht]
\centering
\caption{Average dwell time in the traffic light region for different gains selection and $\phi = 0.25$.}
\label{tab:different gains}
\begin{tabular}{|l|c|c|c|}
\hline
\textbf{Gains Selection} & \textbf{Average time}  \\
\hline
$\kappa_R=0.03, \kappa_T=0.1, \kappa_I=0.02$    & 57.71 s            \\
$\kappa_R=0.04, \kappa_T=0.2, \kappa_I=0.05$       & 56.65 s             \\
$\kappa_R=0.06, \kappa_T=0.3, \kappa_I=0.075$    & 55.76 s             \\
\hline
\end{tabular}
\end{table}

\begin{table}[ht]
\centering
\caption{Average dwell time in the traffic light region for different $\phi$ with $\kappa_R=0.04, \kappa_T=0.2, \kappa_I=0.05$.}
\label{tab:different phis}
\begin{tabular}{|l|c|c|c|}
\hline
\textbf{Gains Selection} & \textbf{Average time}  \\
\hline
$\phi = 0.05$    & 64.86 s            \\
$\phi = 0.25$       & 56.65 s             \\
$\phi = 0.55$    & 55.49 s             \\
\hline
\end{tabular}
\end{table}
\begin{table}[ht]
\centering
\caption{Average dwell time in the traffic light region for different penetration rates and $\kappa_R=0.04, \kappa_T=0.2, \kappa_I=0.05, \phi = 0.25$.}
\label{tab:different penetrations}
\begin{tabular}{|l|c|c|c|}
\hline
\textbf{Penetration rate} & \textbf{Average time}  \\
\hline
CAVs 0\%    & 60.45 s            \\
CAVs 20\%    & 59.03 s            \\
CAVs 40\%       & 57.28 s             \\
CAVs 60\%    & 56.66 s             \\
\hline
\end{tabular}
\end{table}

% \textcolor{red}{Here we combine the signalized + unsignalzied results in one big simulation. The signalized results are for turning left (Green arrow) and crossing straight (green light). The merging rules are for right on red and left during a green light (no green or red arrow). We can also throw in a bunch of IDM or even human driven cars and see how the system reacts.}

\section{Conclusion} \label{sec:conclusion}

In this paper, we presented a comprehensive framework to address the challenges of integrating connected and automated vehicles (CAVs) and human-driven vehicles (HDVs) in mixed traffic scenarios, focusing on long-duration autonomy. By transforming infinite-horizon optimal control problems into simple reactive controllers, we demonstrated a scalable solution for CAV trajectory optimization. We further incorporated control barrier functions to guarantee safety in scenarios involving traffic lights and complex intersections, including right turns on red signals. Our simulations validated the effectiveness of this framework, showing improved traffic throughput and safety. Future research could extend this work by exploring dynamic traffic conditions, and multi-agent interactions to fully realize the potential of autonomous mobility systems.

\bibliographystyle{unsrt}
\bibliography{my_pubs,mendeley,refs}

\vspace{5pt}
\begin{wrapfigure}{l}{0.9in}
\includegraphics[width=1in,height=1.25in,clip,keepaspectratio]{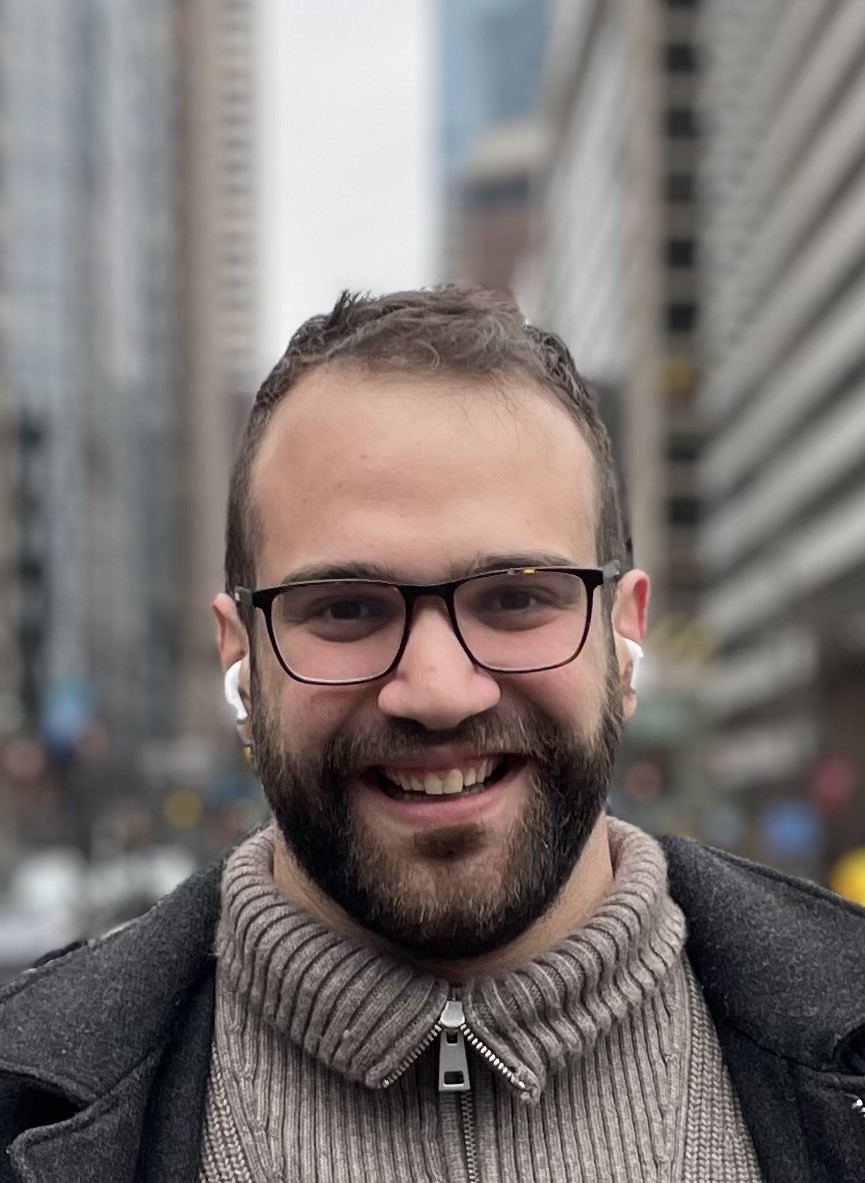}
\end{wrapfigure}
\noindent
\textbf{Filippos N. Tzortzoglou} (Student Member, IEEE) received the Diploma with an integrated M.S. degree in production engineering and management from the Technical University of Crete, Chania, Greece, in 2022. He is currently pursuing the Ph.D. degree with the Civil and Environmental Engineering Department, Cornell University, Ithaca, NY, USA. In 2022, he joined the Mechanical Engineering Department, University of Delaware, Newark, DE, USA, as a Research and Teaching Assistant. Also, in 2024, he joined MathWorks, Natick, MA, USA, as a Research Intern. His research interests lie in the area of automatic control, with applications in transportation and autonomous vehicles. Mr. Tzortzoglou has received several fellowships from foundations across the USA and the National Research Foundation.

\vspace{5pt}

\begin{wrapfigure}{l}{0.9in}
\includegraphics[width=1in,height=1.25in,clip,keepaspectratio]{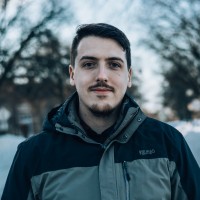}
\end{wrapfigure}

\textbf{Logan E. Beaver} (Member, IEEE) received the B.S. degree from the Milwaukee School of Engineering, Milwaukee, WI, USA, in 2015, and the M.S. degree from Marquette University, Milwaukee, WI, USA, in 2017, both in mechanical engineering, and the Ph.D. degree in mechanical engineering from the University of Delaware, Newark, DE, USA, in 2022.
He was a Postdoctoral Researcher with the Division of Systems Engineering, Boston University, Boston, MA, USA, from 2022 to 2023.
He is currently an Assistant Professor with Old Dominion University, Norfolk, VA, USA, where he directs the Intelligent Systems Laboratory.
His research interests include the stabilization and control of complex multi-agent and swarm systems using optimization and optimal control.

\end{document}